\documentclass[11pt]{article}
\usepackage{fullpage}
\usepackage{amsmath,amsthm,amssymb,amsfonts}
\usepackage{color}
\usepackage{graphicx} 
\usepackage[shortend]{algorithm2e}
\usepackage{appendix}
\usepackage{comment}
\usepackage{tasks}
\usepackage{enumitem}
\usepackage{thm-restate}

\usepackage{ifthen}
\usepackage{tikz}
\usetikzlibrary{calc}
\usetikzlibrary{intersections}
\usetikzlibrary{shapes.geometric}
\newcommand{\writelabel}[2]{
    \node (TT) at (#1) {#2};
}

\newcommand{\drawborder}[2]{
    \draw[rounded corners=4pt, color=gray] (#1) rectangle (#2);
}

\newcommand{\drawplotsegment}[9]{
    \draw [dashed, gray] ($(#1)+(#5,0)$) -- ($(#1)+(#5,#8)$);
    \node (P) at ($(#1)+(#5,#8+0.6)$) {$#9$};
    
    \draw [color=#2] ($(#1)+(#4,#3)$) -- ($(#1)+(#5,#3)$);
    \ifthenelse{\isodd{#6}}{\node [full-node, color=#2] (S) at ($(#1)+(#4,#3)$) {};}{}
    \ifthenelse{\isodd{#7}}{\node [empty-node, draw=#2, fill=white] (E) at ($(#1)+(#5,#3)$) {};}{}
}

\newcommand{\drawsimpleplotsegment}[7]{
    \draw [color=#2] ($(#1)+(#4,#3)$) -- ($(#1)+(#5,#3)$);
    \ifthenelse{\isodd{#6}}{\node [full-node, color=#2] (S) at ($(#1)+(#4,#3)$) {};}{}
    \ifthenelse{\isodd{#7}}{\node [empty-node, draw=#2, fill=white] (E) at ($(#1)+(#5,#3)$) {};}{}
}
\newenvironment{timers}[0]{
    \begin{center}
    \begin{scriptsize}
    \begin{tikzpicture}[
        scale=0.5,
        request-node/.style={fill, diamond, inner sep=0pt, minimum size=4pt, color=trd_main},
    ]
}{
    \end{tikzpicture}
    \end{scriptsize}
    \end{center}
}

\newcommand{\drawarrivaltimeline}[9]{
    \foreach \start/\end/\drawstart in {#9}{
        \draw [rounded corners=2pt, fill=#8, draw=none] ($(#1)+(\start,0.2)$) rectangle ($(#1)+(\end,-0.2)$);
    }

    \draw ($(#1)+(0,-0.15)$) -- ($(#1)+(0,0.15)$);
    \node (L) at ($(#1)+(#2+0.5,0)$) {$#3$};
    
    \foreach \mark/\arrivaltime/\labtype/\point/\number in {#4}{
        \ifthenelse{
            \isodd{\mark}
        }{
            \node [request-node] (M) at ($(#1) + (\arrivaltime,0)$) {};
        }{
            \draw [trd_main, thick] ($(#1)+(\arrivaltime,-0.12)$) -- ($(#1)+(\arrivaltime,0.12)$);
        }
        \node (R) at ($(#1) + (\arrivaltime,0.6)$) {};
    }
    
    \foreach \start/\end/\varname/\varid/\rtype/\rpoint/\rnumber/\isdone in {#5}{
        \ifthenelse{
            \isodd{\isdone}
        }{
            \draw [stealth-stealth] ($(#1)+(\start,0)$) -- ($(#1)+(\end,0)$) node [midway, above] {\ifthenelse{\equal{\rtype}{0}}{$\varname_{\varid}$}{\ifthenelse{\equal{\rtype}{1}}{$\varname_{\varid}^{\rnumber}$}{$\varname_{\varid}^{\rpoint,\rnumber}$}}};
        }{
            \draw [stealth-] ($(#1)+(\start,0)$) -- ($(#1)+(\end,0)$) node [midway, above] {\ifthenelse{\equal{\rtype}{0}}{$\varname_{\varid}$}{\ifthenelse{\equal{\rtype}{1}}{$\varname_{\varid}^{\rnumber}$}{$\varname_{\varid}^{\rpoint,\rnumber}$}}};
            \draw ($(#1)+(\end,-0.12)$) -- ($(#1)+(\end,0.12)$);
        }
    }
    
    \draw [stealth->] ($(#1)+(#6,0)$) -- ($(#1)+(#2,0)$);
    
    \foreach \start/\end/\drawstart in {#9}{
        \ifthenelse{
            \isodd{\drawstart}
        }{
            \draw [stealth-stealth, #7] ($(#1)+(\start,0)$) -- ($(#1)+(\end,0)$);
        }{
            \draw [-stealth, #7] ($(#1)+(\start,0)$) -- ($(#1)+(\end,0)$);
        }
    }
    
    \node (D) at ($(#1) + (#2-0.7,0.6)$) {$\cdots$};
}

\newcommand{\markrequests}[4]{
    \draw [->] (#1) -- +(#2,0);
    \draw ($(#1)+(0,-0.15)$) -- ($(#1)+(0,0.15)$);
    \node (SV) at ($(#1)+(#2+0.5,0)$) {$#3$};
    
    \foreach \mark/\arrivaltime/\labtype/\point/\number in {#4}{
        \ifthenelse{
            \isodd{\mark}
        }{
            \node [request-node] (M) at ($(#1) + (\arrivaltime,0)$) {};
        }{
            \draw [trd_main, thick] ($(#1)+(\arrivaltime,-0.12)$) -- ($(#1)+(\arrivaltime,0.12)$);
        }
        \node (R) at ($(#1) + (\arrivaltime,0.6)$) {\ifthenelse{\isodd{\labtype}}{$r_{\point}^{\number}$}{$r_{\number}$}};
    }
}

\newcommand{\drawgrayarrivaltimeline}[5]{
    \draw [->, gray] (#1) -- +(#2,0);
    \draw [gray] ($(#1)+(0,-0.15)$) -- ($(#1)+(0,0.15)$);
    \node (SV) at ($(#1)+(#2+0.5,0)$) {$#3$};
    
    \foreach \mark/\arrivaltime in {#4}{
        \ifthenelse{
            \isodd{\mark}
        }{
            \node [request-node] (M) at ($(#1) + (\arrivaltime,0)$) {};
        }{
            \draw [trd_main, thick] ($(#1)+(\arrivaltime,-0.12)$) -- ($(#1)+(\arrivaltime,0.12)$);
        }
        \node (R) at ($(#1) + (\arrivaltime,0.6)$) {};
    }
    
    \foreach \start/\end/\varname/\varnum in {#5}{
        \draw [stealth-stealth] ($(#1)+(\start,0)$) -- ($(#1)+(\end,0)$) node [midway, above] {$\varname_{\varnum}$};
    }
    
    \node (D) at ($(#1) + (#2-0.7,0.6)$) {$\cdots$};
}

\newcommand{\drawreferencelines}[3]{
    \foreach \arrival in {#3}{
        \draw [dotted, gray] ($(#1)+(\arrival,0)$) -- ($(#1)+(\arrival,#2)$);
    }
}

\newcommand{\markrequestsandvariables}[6]{
    \draw ($(#1)+(0,-0.15)$) -- ($(#1)+(0,0.15)$);
    \node (SV) at ($(#1)+(#2+0.5,0)$) {$#3$};
    
    \foreach \mark/\arrivaltime/\number in {#4}{
        \ifthenelse{
            \isodd{\mark}
        }{
            \node [request-node] (M) at ($(#1) + (\arrivaltime,0)$) {};
        }{
            \draw [trd_main, thick] ($(#1)+(\arrivaltime,-0.12)$) -- ($(#1)+(\arrivaltime,0.12)$);
        }
        \node (R) at ($(#1) + (\arrivaltime,-0.6)$) {$r_{\number}$};
    }
    
    \foreach \start/\end/\varname/\rtype/\varid/\rid in {#5}{
        \draw [stealth-stealth] ($(#1)+(\start,0)$) -- ($(#1)+(\end,0)$);
        \node (L) at ($(#1)+(\start+1/2*\end-1/2*\start,0.5)$) {\ifthenelse{\isodd{\rtype}}{$\varname_{\varid}^{\rid}$}{$\varname_{\varid}$}};
    }
    
    \draw [stealth->] ($(#1)+(#6,0)$) -- ($(#1)+(#2,0)$);
    
    \node (D) at ($(#1) + (#2-0.7,0.6)$) {$\cdots$};
}

\newcommand{\drawstoppedarrivaltimeline}[9]{
    \foreach \start/\end/\drawstart in {#9}{
        \draw [rounded corners=2pt, fill=#8, draw=none] ($(#1)+(\start,0.2)$) rectangle ($(#1)+(\end,-0.2)$);
    }

    \draw ($(#1)+(0,-0.15)$) -- ($(#1)+(0,0.15)$);
    
    \foreach \mark/\arrivaltime/\labtype/\point/\number in {#4}{
        \ifthenelse{
            \isodd{\mark}
        }{
            \node [request-node] (M) at ($(#1) + (\arrivaltime,0)$) {};
        }{
            \draw [trd_main, thick] ($(#1)+(\arrivaltime,-0.12)$) -- ($(#1)+(\arrivaltime,0.12)$);
        }
        \node (R) at ($(#1) + (\arrivaltime,0.6)$) {};
    }
    
    \foreach \start/\end/\varname/\varid/\rtype/\rpoint/\rnumber/\isdone in {#5}{
        \ifthenelse{
            \isodd{\isdone}
        }{
            \draw [stealth-stealth] ($(#1)+(\start,0)$) -- ($(#1)+(\end,0)$) node [midway, above] {\ifthenelse{\equal{\rtype}{0}}{$\varname_{\varid}$}{\ifthenelse{\equal{\rtype}{1}}{$\varname_{\varid}^{\rnumber}$}{$\varname_{\varid}^{\rpoint,\rnumber}$}}};
        }{
            \draw [stealth-] ($(#1)+(\start,0)$) -- ($(#1)+(\end,0)$) node [midway, above] {\ifthenelse{\equal{\rtype}{0}}{$\varname_{\varid}$}{\ifthenelse{\equal{\rtype}{1}}{$\varname_{\varid}^{\rnumber}$}{$\varname_{\varid}^{\rpoint,\rnumber}$}}};
            \draw ($(#1)+(\end,-0.12)$) -- ($(#1)+(\end,0.12)$);
        }
    }
    
    \foreach \start/\end/\drawstart in {#9}{
        \ifthenelse{
            \isodd{\drawstart}
        }{
            \draw [stealth-stealth, #7] ($(#1)+(\start,0)$) -- ($(#1)+(\end,0)$);
        }{
            \draw [-stealth, #7] ($(#1)+(\start,0)$) -- ($(#1)+(\end,0)$);
        }
    }
}
\definecolor{fst}{RGB}{105,190,40}
\colorlet{fst_main}{fst!90!black}
\colorlet{fst_fill}{fst!10}

\definecolor{snd}{RGB}{13,152,186}
\colorlet{snd_main}{snd!90!black}
\colorlet{snd_fill}{snd!10}

\definecolor{add}{RGB}{177,148,216}
\colorlet{add_main}{add!90!black}
\colorlet{add_fill}{add!10}

\definecolor{trd}{RGB}{255,210,0}
\colorlet{trd_main}{trd!90!black}

\def\P{\mathbb{P}}
\def\R{\mathbb{R}}
\def\E{\mathbb{E}}

\def\delay{\texttt{delay}}
\def\weight{\texttt{weight}}
\def\cost{\texttt{cost}}
\def\llambda{\boldsymbol{\lambda}}

\newcommand{\br}[1]{\left(#1\right)}

\newcommand{\opt}{{\rm OPT}}

\newcommand{\alg}{{\rm ALG}}

\newcommand{\mlaroot}{\gamma}
\newcommand{\treeroot}[1]{\gamma\br{#1}}
\newcommand{\parent}[1]{\text{par}\!\br{#1}}

\newcommand{\plan}{{\rm PLAN}}
\newcommand{\instant}{{\rm INSTANT}}
\newcommand{\general}{{\rm GEN}}
\newcommand{\gen}{{\rm GEN}}

\newcommand{\expdistr}[1]{{\rm Exp}(#1)}

\newtheorem{definition}{Definition}[section]
\newtheorem{theorem}[definition]{Theorem}
\newtheorem{lemma}[definition]{Lemma}
\newtheorem{observation}[definition]{Observation}
\newtheorem{proposition}[definition]{Proposition}
\newtheorem{corollary}[definition]{Corollary}

\newtheorem{claim}[definition]{Claim}

\allowdisplaybreaks[4]

\title{Online Multi-level aggregation with\\ delays and stochastic arrivals}
\date{}
\author{Mathieu Mari\thanks{LIRMM, University of Montpellier, Montpellier, France. mari.mathieu.06@gmail.com}, Michał Pawłowski\thanks{MIMUW, University of Warsaw and IDEAS NCBR, Warsaw, Poland. michal.pawlowski196@gmail.com}, Runtian Ren\thanks{University of Wrocław and IDEAS NCBR, Wrocław, Poland. renruntian@gmail.com}, Piotr Sankowski\thanks{MIMUW, University of Warsaw, IDEAS NCBR, MIM Solutions, Warsaw, Poland. piotr.sankowski@gmail.com}}

\begin{document}

\maketitle

\begin{abstract}
This paper presents a new research direction for online Multi-Level Aggregation (MLA) with delays. 
In this problem, we are given an edge-weighted rooted tree $T$, and we have to serve a sequence of requests arriving at its vertices in an online manner. 
Each request $r$ is characterized by two parameters: its arrival time $t(r)$ and location $l(r)$ (a vertex). 
Once a request $r$ arrives, we can either serve it immediately or postpone this action until any time $t > t(r)$. 
We can serve several pending requests at the same time, and the service cost of a service corresponds to the weight of the subtree that contains all the requests served and the root of $T$. 
Postponing the service of a request $r$ to time $t > t(r)$ generates an additional delay cost of $t - t(r)$. 
The goal is to serve all requests in an online manner such that the total cost (i.e., the total sum of service and delay costs) is minimized. 
The current best algorithm for this problem achieves a competitive ratio of $O(d^2)$ (Azar and Touitou, FOCS'19), where $d$ denotes the depth of the tree. 

The MLA problem is a generalization of several well-studied problems, including TCP Acknowledgment (depth 1), Joint Replenishment (depth 2) and multi-level message aggregation (arbitrary depth). 
Although it appeared implicitly in many previous papers, it has been formalized by Bienkowski et al.~(ESA'16).

Here, we consider a stochastic version of MLA where the requests follow a Poisson arrival process. 
We present a deterministic online algorithm which achieves a constant ratio of expectations, meaning that the ratio between the expected costs of the solution generated by our algorithm and the optimal offline solution is bounded by a constant. 
Our algorithm is obtained by carefully combining two strategies. 
In the first one, we plan periodic oblivious visits to the subset of frequent vertices, whereas in the second one, we greedily serve the pending requests in the remaining vertices. 
This problem is complex enough to demonstrate a very rare phenomenon that ``single-minded" or ``sample-average" strategies are not enough in stochastic optimization. 
\end{abstract}

\section{Introduction}
\label{section:intro}
Imagine the manager of a biscuit factory needs to deal with the issue of delivering products from the factory to the convenience stores. 
Once some products, say chocolate waffle, is in shortage at some store, then the store employee will inform the factory for replenishment.  
From the factory's perspective, each time a service is created to deliver the products, a truck has to travel from the factory to go to each store, and then come back to the factory.
A cost proportional to the total traveling distance has to be paid for this service. 
For the purpose of saving delivery cost, it is beneficial to accumulate the replenishment requests from many stores and then deliver the ordered products altogether in one service.  
However, this accumulated delay of delivering products may cause the stores unsatisfied and the complaints will have negative influence on future contracts between the stores and the factory.
Typically, for each request ordered from a store, the time gap between ordering the products and receiving the products, is known as delay cost (of this request).
The goal of the factory manager, is to plan the delivery service schedule in an online manner, such that the total service cost and the total delay cost is minimized. 

The above is an example of an online problem called Multi-level Aggregation (MLA) with linear delays. 
Formally, the input is an edge-weighted rooted tree $T$ and a sequence of requests, with each request $r$ specified by an arrival time $t(r)$ and a location at a particular vertex.
Once a request $r$ arrives, its service does not have to be processed immediately, but can be delayed to any time $t \ge t(r)$ at a delay cost of $t - t(r)$.
The benefit of delaying requests is that several requests can be served together to save some service cost: to serve any set of requests $R$ at time $t$, a subtree $T'$ containing the tree root and all locations of requests $R$ needs to be bought at a service cost equal to the total weight of edges in $T'$. 
The goal of MLA is to serve all requests in an online manner such that the total cost (i.e., the total service cost plus the total delay cost) is minimized. 

The MLA problem is first formally introduced by Bienkowski et al. \cite{bienkowski2016online}. 
Due to many real-life applications ranging from logistic, supply chain management, data transmission in sensor network, this MLA problem has recently drawn considerable attentions \cite{bienkowski2013online, bienkowski2016online, buchbinder2017depth, azar2019general}. 
Besides, two classic problems, TCP-acknowledgment (also known as lot-sizing problem, from operation research community) and Joint Replenishment (JRP), as special cases of MLA with tree depths of 1 and 2 respectively, are studied by extensive previous works \cite{dooly2001line, karlin2001dynamic, seiden2000guessing, aggarwal1993improved, khanna2002control, buchbinder2008online, bienkowski2014better, arkin1989computational, nonner2009approximating, bienkowski2015approximation}.
Particularly for MLA, the state-of-the-art is as follows:
\begin{itemize}
    \item[-] the current best online algorithm, proposed by Azar and Touitou \cite[Theorem IV.2]{azar2019general}, achieves a competitive ratio of O($d^2$), where $d$ denotes the depth of the given tree (i.e., the maximum number of edges from any leaf vertex to the tree root);
    \item[-] no online algorithm can achieve a competitive ratio less than 4 \cite[Theorem 6.3]{bienkowski2016online} --- this is the current best lower bound, even restricted to the case when the given tree is a path, and the root is an endpoint of the path. 
\end{itemize}
Obviously, there is a huge gap between the upper bound and the lower bound on the competitiveness of MLA. 
Closing the gap remains an interesting open question. 

In fact, it is often too pessimistic to assume no stochastic information on the input is available in practice --- again, consider our delivery example. 
The factory knows all the historical orders and can estimate the request frequencies from the stores of all locations. 
It is reasonable to assume that the requests follow some stochastic distribution. 
Therefore, the following question is natural: {\em if stochastic information on the input is available, can we devise online algorithms for MLA with better performance guarantees?}

In this paper, we provide an affirmative answer to this question. 
We study a stochastic online version of MLA, assuming that the requests arrive following a Poisson arrival process. 
More precisely, the waiting time between any two consecutive requests arriving at the same vertex $u$ follows an exponential distribution $\expdistr{\llambda(u)}$ with parameter $\llambda(u)$.
In this model, the goal is to minimize the expected cost produced by an algorithm $\alg$ for a random input sequence generated in a long time interval $[0, \tau]$.
In order to evaluate the performance of our algorithms on stochastic inputs, we use the {\em ratio of expectations} (RoE), that corresponds to the ratio of the expected cost of the algorithm to the expected cost of the optimal offline solution (see Definition \ref{def:roe}). 

\paragraph{Our contribution.}
We prove that the performance guarantee obtained in the Poisson arrival model is significantly better compared with the current best competitiveness obtained in the adversarial model.
More specifically, we propose a non-trivial deterministic online algorithm which achieves a constant ratio of expectations.
\begin{theorem} \label{theorem:main_intro}
    For MLA with linear delays in the Poisson arrival model, there exists a deterministic online algorithm which achieves a constant ratio of expectations. 
\end{theorem}

Our algorithm is obtained by synergistically merging two carefully crafted strategies. 
The first strategy incorporates periodic oblivious visits to a subset of frequently accessed vertices, while the second strategy employs a proactive, greedy approach to handle pending requests in the remaining vertices.  
The complexity of this problem unveils a rare phenomenon --- the inadequacy of ``single-minded" or ``sample-average" strategies in stochastic optimization. 
In this paper we not only address this challenge but also point on to further complex problems that require similar approach in stochastic environments. 
We stress that it is open to obtain improved results for stochastic cases of facility location with delays \cite{azar2019general} or online service with delays \cite{azar2017online}.

\paragraph{Previous works. }
The MLA problem has been only studied in the adversarial model. 
Bienkowski et al.~\cite{bienkowski2016online} introduced a general version of MLA, assuming that the cost of delaying a request $r$ by a duration $t$ is $f_r(t)$.
Here $f_r(\cdot)$ denotes the delay cost function of $r$ and it only needs to be non-decreasing and satisfy $f_r(0) = 0$.
They proposed an O($d^4 2^d$)-competitive online algorithm for this general delay cost version problem, where $d$ denotes the tree depth \cite[Theorem 4.2]{bienkowski2016online}. 
Later, the competitive ratio is further improved to O($d^2$) by Azar and Touitou \cite[Theorem IV.2]{azar2019general} (for this general delay cost version). 
However, no matching lower bound has been found for the delay cost version of MLA --- the current best lower bound on MLA (with delays) is only 4 \cite[Theorem 6.3]{bienkowski2016online}, restricted to a path case with linear delays. 
Thus far, no previous work has studied MLA in the stochastic input model.

\paragraph{Organization. }
In Section \ref{section:preliminaries}, we give all the necessary notations and preliminaries. 
In Section \ref{section:warmup}, we study a special single-edge tree instance as a warm-up.
We show that there are two different situations, one heavy case and one light case, and to achieve constant RoE, the ideas for the two cases are different.
In Section \ref{section:overview}, we give an overview of our deterministic online algorithm (Theorem \ref{theorem:main_intro}). 
This algorithm is the combination of two different strategies for two different types of instances. In Section \ref{section:lightcase}, we study the first type, called light instances, that are a generalization of light single-edge trees. 
In Section \ref{section:heavycase}, we study the other type called heavy instances as a generalization of heavy single-edge trees. 
In Section \ref{section:general}, we prove Theorem \ref{theorem:main_intro}. 
We finish the paper by detailing in Section \ref{section:related} all other related works, and by discussing some future directions in Section \ref{section:conclusion}. 

\section{Notations and Preliminaries}
\label{section:preliminaries}
\paragraph{Weighted tree.} 
Consider an edge-weighted tree $T$ rooted at vertex $\treeroot{T}$. 
We refer to its vertex set by $V(T)$ and its edge set by $E(T)$. 
When the context is clear, we denote the root vertex, vertex set, and edge set by $\mlaroot$, $V$, and $E$, respectively. 
We assume that each edge $e \in E$ has a positive weight $w_e$. 
For any vertex $u \in V$, except for the root vertex $\mlaroot$, we denote its \emph{parent vertex} as $\parent{u} \in V$, and $e_u = (u,\parent{u})$ as the edge connecting $u$ and its parent. 
We also define $T_u$ as the subtree of $T$ rooted at vertex $u$. 
In addition to the edge weights, we use the term \emph{vertex weight} to refer to $w_u := w(e_u)$, where $u \in V$ and $u \neq \mlaroot$.
Given any two vertices $u, v \in V(T)$, we denote the path length from $u$ to $v$ in $T$ by $d_T(u, v)$, i.e., it is the total weight of the edges along this path.\footnote{When the context is clear, we simply write $d(u, v)$ instead of $d_T(u, v)$. Furthermore, we stress that the order of vertices in this notation is not arbitrary --- the second vertex ($v$) is always an ancestor of the first one ($u$).}
Finally, we use $T[U]$ to denote the forest induced by vertices of $U \subseteq V(T)$ in $T$. 

\paragraph{Problem description.} 
An MLA \emph{instance} is characterized by a tuple $(T, \sigma)$, where $T$ is a weighted tree rooted at $\gamma$ and $\sigma$ is a sequence of \emph{requests}.
Each request $r$ is described by a tuple $(t(r), l(r))$ where $t(r) \in \mathbb{R}^+$ denotes $r$'s \emph{arrival time} and $l(r) \in V(T)$ denotes $r$'s \emph{location}. 
Thus, denoting by $m$ the number of requests, we can rewrite $\sigma := (r_1, \dots, r_m)$ with the requests sorted in increasing order of their arrival times, i.e., $t(r_1) \le t(r_2) \le \dots \le t(r_m)$.
Given a sequence of requests $\sigma$, a \emph{service} $s = (t(s), R(s))$ is characterized by the \emph{service time} $t(s)$ and the set of requests $R(s) \subseteq \sigma$ it serves. 
A \emph{schedule} $S$ for $\sigma$ is a sequence of services. 
We call schedule $S$ \emph{valid} for $\sigma$ if each request $r \in \sigma$ is assigned a service $s \in S$ that does not precede $r$'s arrival.
In other words, a valid $S$ for $\sigma$ satisfies
(i) $\forall \, s \in S$ $\forall \, r \in R(s)$ $t(r) \ge t(s)$; 
(ii) $\{R(s): s \in S\}$ forms a partition of $\sigma$.
Given any MLA instance ($T, \sigma$), an MLA algorithm $\alg$ needs to produce a valid schedule $S$ to serve all the requests $\sigma$.
Particularly for an online MLA algorithm $\alg$, at any time $t$, the decision to create a service to serve a set of pending request(s) cannot depend on the requests arriving after time $t$.

For each request $r \in \sigma$, let $S(r)$ denote the service in $S$ which serves $r$, i.e., for each $s \in S$, $S(r) = s$ if and only if $r \in R(s)$. 
Given a sequence of requests $\sigma$ and a valid schedule $S$, the \emph{delay} cost for a request $r \in \sigma$ is defined as $\delay(r) := t(S(r)) - t(r)$. 
Using this notion, we define the \emph{delay} cost for a service $s \in S$ and the \emph{delay} cost for the schedule $S$ as
\begin{equation*}
    \delay(s) := \sum_{r \in R(s)} \delay(r) \qquad\quad \text{and} \qquad\quad \delay(S) := \sum_{s \in S} \delay(s).
\end{equation*}
Besides, given any request $r \in \sigma$, if it is pending at time $t > t(r)$, let $\delay(r, t) = t - t(r)$ denote the delay cost of $r$ at this moment.

The \emph{weight} (also called \emph{service cost}) of a service $s \in S$, denoted by $\weight(s, T)$, is defined as the weight of the minimal subtree of $T$ that contains root $\mlaroot$ and all locations of requests $R(s)$ served by $s$. 
The \emph{weight} (or \emph{service cost}) of a schedule $S$ is defined as $\weight(S, T) := \sum_{s \in S} \weight(s, T)$.
To compute the \emph{cost} of a service $s$, we sum its delay cost and weight, i.e., 
\begin{equation*}
    \cost(s, T) := \delay(s) + \weight(s, T).
\end{equation*}
Similarly, we define the \emph{cost} (or \emph{total cost}) of a schedule $S$ for $\sigma$ as 
\begin{equation*}
    \cost(S, T) := \delay(S) + \weight(S, T).
\end{equation*}
When the context is clear, we simply write $\cost(S) = \cost(S,T)$. 
Moreover, given an MLA instance $(T, \sigma)$, let $\alg(\sigma)$ denote the schedule of algorithm $\alg$ for $\sigma$ and let $\opt(\sigma)$ denote the optimal schedule for $\sigma$ with minimum total cost. 
Note that without loss of generality, we can assume that no request in $\sigma$ arrives at the tree root $\gamma$ since such a request can be served immediately at its arrival with zero cost.

\paragraph{Poisson arrival model.} 
In this paper, instead of using an adversarial model, we assume that the requests arrive according to some stochastic process. 
A \emph{stochastic instance} that we work with is characterized by a tuple $(T,\llambda)$, where $T$ denotes an edge-weighted rooted tree, and $\llambda: V(T) \rightarrow \R_+$ is a function that assigns each vertex $u \in V(T)$ an \emph{arrival rate} $\llambda(u) \ge 0$.\footnote{Without loss of generality we assume $\llambda(\gamma(T)) = 0$, i.e., no request arrives at the tree root.} 
Formally, such a tuple defines the following process.

\begin{definition}[Poisson arrival model] \label{def:poisson_model_distributed}
    Given any stochastic MLA instance $(T, \llambda)$ and any value $\tau > 0$, we say that a (random) \emph{requests sequence} $\sigma$ follows a \emph{Poisson arrival model} over time interval $[0, \tau]$, if  
    (i) for each vertex $u \in V(T)$ with $\llambda(u) > 0$ the \emph{waiting time} between any two consecutive requests arriving at $u$ follows an \emph{exponential distribution} with parameter $\llambda(u)$;\footnote{For the first request $r$ arriving at $u$, we require that the waiting time from 0 to $t(r)$ follows this distribution $\expdistr{\llambda(u)}$. Similarly, if we look at the last request $r'$ arriving at $u$ and let $W_{r'} \sim \expdistr{\llambda(u)}$ denote the variable determining its waiting time, we require that $\tau - t(r') < W_{r'}$.} 
    (ii) variables representing waiting times are mutually independent;
    (iii) all the requests in $\sigma$ arrive within time interval $[0, \tau]$.
    We denote this fact by writing $\sigma \sim (T,\llambda)^{\tau}$.
\end{definition}

\noindent 
Given any subtree $T'$ of $T$, we use both $\llambda|_{T'}$ and $\llambda|_{V(T')}$ to denote the arrival rates restricted to the vertices of $T'$. 
Similarly, given a random sequence of requests $\sigma \sim (T, \llambda)^{\tau}$, we use $\sigma|_{T'} \subseteq \sigma$ and $\sigma|_{I} \subseteq \sigma$ for $I \subseteq [0, \tau]$ to denote the sequences of all requests in $\sigma$ that arrive inside the subtree $T'$ and within the time interval $I$, respectively. 
Note that the above Poisson arrival model satisfies the following properties (see Appendix \ref{appendix:preliminaries} for the formal proof). 

\begin{proposition} \label{prop:poisson_independence_under_taking_subtree}
    Given a subtree $T'$ of $T$: 
    (i) for any $\tau > 0$ and any sequence $\sigma \sim (T, \llambda)^\tau$, $\sigma|_{T'}$ follows the Poisson arrival model over the MLA instance restricted to $T'$, i.e, $\sigma|_{T'} \sim (T', \llambda|_{T'})^{\tau}$,
    (ii) the process determining arrivals inside $T'$ is independent of the requests arriving in $T \setminus T'$.
\end{proposition}

\begin{proposition} \label{prop:poisson_independence_under_taking_subsegment}
    Given any stochastic MLA instance $(T, \llambda)$, $\tau = \sum_{i=1}^k \tau_i$ with $\tau_i > 0$ for each $i \in [k]$,\footnote{For simplicity, we use $[k]$ to denote $\{1, 2, \dots, k\}$ everywhere in this paper. }
    and a family of random sequences $\{\sigma_i \sim (T,\llambda)^{\tau_i}: i \in [k]\}$, we merge them into one sequence defined over a $\tau$-length time interval by postponing arrivals of requests in $\sigma_i$ by $\sum_{j=1}^{i-1} \tau_j$ for all $i \in [k]$. 
    Due to the memoryless property of exponential variables, this process results in a \emph{sequence $\sigma$ that follows the Poisson arrival model} over $[0, \tau]$, i.e., $\sigma \sim (T,\llambda)^{\tau}$.
\end{proposition}

\noindent
Intuitively, the first proposition gives us the freedom to select a sub-instance of the problem and focus on the requests arriving in a subtree $T' \subseteq T$. 
The second one allows us to split the time horizon into smaller intervals and work with shorter request sequences. 
The most important fact, though, is that both operations preserve the arrival model and are independent of the remaining part of the initial request sequence.  
Here, we also stress that in the following sections, we use the notation of $\llambda(T') := \sum_{v \in T'} \llambda(v)$ to denote the arrival rate for a given subtree $T' \subseteq T$.

Another equivalent characteristic of the Poisson arrival model gives us a more ``centralized" perspective on how the request sequences are generated (see Appendix \ref{appendix:preliminaries} for the formal proof). 

\begin{proposition} \label{prop:poisson_model_centralized}
    Given any stochastic MLA instance $(T, \llambda)$ and a random sequence of requests $\sigma \sim (T, \llambda)^\tau$, we have 
    (i) the waiting time between any two consecutive requests in $\sigma$ follows an exponential distribution with parameter $\llambda(T)$; 
    (ii) for each vertex $u \in T$ and each request $r \in \sigma$ the probability of $r$ being located at $u$ equals $\llambda(u) / \llambda(T)$. 
\end{proposition}

In the following, we introduce three more properties of the Poisson arrival model. 
To simplify their statements, from now on we denote the random variable representing the number of requests in sequence $\sigma \sim (T, \llambda)^{\tau}$ by $N(\sigma)$. 
The first property describes the expected value of $N(\sigma)$ for a fixed time horizon $\tau$. 
The second one describes our model's behavior under the assumption that we are given the value of $N(\sigma)$. 
Finally, the third one presents the value of the expected waiting time generated by all the requests arriving before a fixed time horizon.
All the proofs can be found in \cite{ross1996stochastic}.\footnote{The proof of the first proposition follows from Proposition 2.2.1, and Definition 2.1.1, while the proofs of the remaining two facts can be found in Theorem 2.3.1 and Example 2.3(A).}
However, for completeness, we also include them in Appendix \ref{appendix:preliminaries}.
\begin{proposition} \label{prop:poisson_expected_number_requests}
    Given any stochastic MLA instance $(T, \llambda)$ and a random sequence of requests $\sigma \sim (T, \llambda)^\tau$, it holds that (i) $N(\sigma) \sim Pois(\llambda(T) \cdot \tau)$; (ii) $\E[N(\sigma) \mid \sigma \sim (T, \llambda)^\tau] = \llambda(T) \cdot \tau$; (iii) if $\llambda(T) \cdot \tau \geq 1$, then $\P(N(\sigma) \geq \E[N(\sigma)]) \geq 1/2$.
\end{proposition}

\begin{proposition} \label{prop:poisson_fixed_number_uniform_arrivals}
    Given that $n$ requests arrive during time interval $[0, \tau]$ according to Poisson arrival model, the $n$ arrival times (in sequence) have the same distribution as the order statistics corresponding to $n$ independent random variables uniformly distributed over $[0, \tau]$. 
\end{proposition}

\begin{proposition} \label{prop:poisson_total_delay_cost}
    Given any stochastic MLA instance $(T, \llambda)$ and $\sigma \sim (T, \llambda)^\tau$, the expected delay cost generated by all the requests arriving before $\tau$ is equal to
    \begin{equation*}
        \E\left[\sum_{i = 1}^{N(\sigma)} (\tau - t(r_i)) \mid \sigma \sim (T, \llambda)^\tau\right] = \frac{\tau}{2} \cdot \E\big[N(\sigma) \mid \sigma \sim (T, \llambda)^\tau\big] = \frac{1}{2} \cdot \llambda(T) \cdot \tau^2.
    \end{equation*}
\end{proposition}

\paragraph{Benchmark description.}
For an online algorithm $\alg$ that takes as input a random sequence $\sigma \sim (T, \llambda)^\tau$, let $\E_\sigma^\tau[\cost(\alg(\sigma), T)]$ denotes the expected cost of the schedule it generates. 
To measure the performance of $\alg$ in this stochastic version of MLA, we use the \emph{ratio of expectations}.

\begin{definition} [ratio of expectations] \label{def:roe}
    An online MLA algorithm $\alg$ achieves a ratio of expectations (RoE) $C \ge 1$, if for all stochastic MLA instances $(T,\llambda)$ we have
    \begin{equation*}
        \overline{\lim_{\tau \to \infty}} \frac{\E[\cost(ALG(\sigma), T) \mid \sigma \sim (T, \llambda)^\tau]}{\E[\cost(OPT(\sigma), T) \mid \sigma \sim (T, \llambda)^\tau]} \le C.
    \end{equation*}
\end{definition}

\section{Warm-up: single edge instances}
\label{section:warmup}
In this section, we study the case of a single-edge tree in the stochastic model. 
Thus, throughout this section, we fix a tree $T$ that consists of a single edge $e = (u, \gamma)$ of weight $w > 0$, and denote the arrival rate of $u$ by $\lambda > 0$. 
In such a setting, the problem of finding the optimal schedule to serve the requests arriving at vertex $u$ is known as \emph{TCP acknowledgment} (here, we consider the stochastic model). 
It is worth mentioning that in the adversarial setting, a 2-competitive deterministic and a $(1-e^{-1})^{-1}$-competitive randomized algorithms are known for this problem \cite{dooly2001line, karlin2001dynamic}.\footnote{To our best extent, no previous work studied this problem in the Poisson arrival model from a theoretical perspective, i.e., evaluating the performance of the algorithms using the ratio of expectations.} 

Let us stress that the goal of this section is not to improve the best-known competitive ratio for a single-edge case, but to illustrate the efficiency of two opposite strategies, and introduce the important concepts of this paper. 
The first strategy called the \emph{instant strategy}, is to serve each request as soon as it arrives. 
Intuitively, this approach is efficient when the requests are not so frequent, so that on average, the cost of delaying a request to the arrival time of the next request, is enough to compensate the service cost. 
The second strategy, called the \emph{periodic approach} is meant to work in the opposite case where requests are frequent enough so that it is worth grouping several of them for the same service. 
In this way, the weight cost of a service can be shared between the requests served. 
Assuming that requests follow some stochastic assumptions, it makes sense to enforce that services are ordered at regular time intervals, where the time between any two consecutive services is a fixed number $p$, which depends only on the instance's parameters. 

There are two challenges here. 
First, when should we use each strategy? 
Second, what should be the value of $p$ that optimizes the performance of the periodic strategy? 
For the first question, we show that this depends on the value of $\pi := w\lambda$ that we call the \emph{heaviness} of the instance. 
More precisely, we show that if $\pi > 1$, i.e., the instance is \emph{heavy}, the periodic strategy is more efficient. 
On the other hand, if $\pi \leq 1$, then the instance is \emph{light}, and the instant strategy is essentially better. 
For the second question, we show that the right value for the period, up to a constant in the ratio of expectations, is $p = \sqrt{2w/\lambda}$. 
Without loss of generality, in what follows we assume that the time horizon $\tau$ is always a multiple of the period chosen, which simplifies the calculation and does not affect the ratio of expectations.

\begin{lemma}
    Given a stochastic instance where the tree consists of a single edge of weight $w > 0$ and the leaf has an arrival rate $\lambda > 0$, let $\pi = w\lambda$ and let $\sigma$ be a random sequence of requests of duration $\tau$, for some $\tau > 0$. 
    It holds that
    \begin{enumerate}[topsep=2pt, itemsep=-2pt, label=(\roman*)]
        \item the instant strategy on $\sigma$ has the expected cost of $\tau \cdot \pi$;
        \item the periodic strategy on $\sigma$, with period $p = \sqrt{2w / \lambda}$, has the expected cost of $\tau \cdot \sqrt{2\pi}$.
    \end{enumerate}
\end{lemma}

\begin{proof}
    Notice that the instant strategy incurs an expected cost equal to the expected number of requests arriving within the time horizon $\tau$ times the cost of serving one. By Proposition \ref{prop:poisson_expected_number_requests}, we have that on average $\lambda \cdot \tau$ requests arrive within the time interval $[0, \tau]$. Thus, since the cost of serving one equals $w$, the total expected cost is $\lambda \cdot \tau \cdot w = \tau \cdot \pi$.

    Similarly, for the periodic strategy, we know that within each period $p = \sqrt{2w / \lambda}$, we generate the expected delay cost of $1/2 \cdot \lambda \cdot p^2 = w$ (Proposition \ref{prop:poisson_total_delay_cost}). The service cost we pay at the end of each period equals $w$ as well. Thus, the total expected cost within $[0, \tau]$ is equal to $\tau / p \cdot 2w = \tau \cdot \sqrt{2\lambda w} = \tau \cdot \sqrt{2\pi}$, which ends the proof.
\end{proof}

We now compare these expected costs with the expected cost of the optimal offline schedule. 
The bounds obtained imply that the instant strategy has constant RoE when $\pi\le 1$, and the periodic strategy (with $p = \sqrt{2w/\lambda}$) has a constant RoE when $\pi > 1$. 

\begin{lemma} \label{lemma:opt_edge}
    Given a stochastic instance where the tree consists of a single edge of weight $w > 0$ and the leaf has an arrival rate $\lambda > 0$, let $\pi = w\lambda$ and let $\sigma$ be a random sequence of requests of duration $\tau$, for some $\tau > 0$. 
    The lower bounds for the optimal offline schedule for $\sigma$ are as follows
    \begin{enumerate}[topsep=2pt, itemsep=-2pt, label=(\roman*)]
        \item if $\pi \leq 1$, then it has an expected cost of at least $1/2 \cdot (1-e^{-1}) \cdot \tau \cdot \pi$;
        \item if $\pi > 1$, then it has an expected cost of at least $3 / 8\sqrt{2} \cdot \tau \cdot \sqrt{\pi}$.
    \end{enumerate}
\end{lemma}

In the following subsection, we prove Lemma \ref{lemma:opt_edge}.

\subsection{Lower bounding OPT}
Let $\sigma \sim (T, \llambda)^\tau$ be a random sequence of requests defined for the given single-edge instance and some time horizon $\tau$. 
In this instance, the edge has weight $w > 0$ and the vertex has arrival rate $\lambda > 0$.
Now we lower bound the expected cost of the optimal offline algorithm OPT on $\sigma$. 

The main idea is to partition the initial time horizon $[0, \tau]$ into a collection of shorter intervals $\{I_1, I_2, \ldots, I_k\}$ of length $p$ each, for some value $p$ that is defined later. 
We denote $\sigma_i := \sigma|_{I_i}$ for $i \in [k]$. 
From Proposition \ref{prop:poisson_independence_under_taking_subsegment}, we know that all $\sigma_i$ are independent and follow the same Poisson arrival model $(T, \llambda)^{p}$.
Thus, we should be able to analyze them separately and combine the results to get the estimation of the total cost incurred by OPT over the initial sequence $\sigma$.

Let $D(\sigma_1)$ denote the total delay cost of $\sigma_1$ at time $p$ when no services are issued during $[0,p]$. 
Note that $\opt$ either serves some requests during $[0, p]$ and incurs the service cost of at least $w$, or issues no services during $[0,p]$ and pays the delay cost of $D(\sigma_1)$. 
Thus, the total cost of $\opt$ within $[0,p]$ is at least $\min(w,D(\sigma_1))$. 
Since $\{\sigma_i: i\in[\tau/p]\}$ are i.i.d, we deduce the following bound:
\begin{equation} \label{eq:edge_opt_1}
    \E\big[\cost(\opt(\sigma))\big] \ge \frac{\tau}{p} \cdot \E\big[\min(w, D(\sigma_1)) \mid \sigma_1 \sim (T, \llambda)^p\big].
\end{equation}

Using Proposition \ref{prop:poisson_fixed_number_uniform_arrivals}, we can partition the right-hand side of the inequality further.
Indeed, we know that when conditioned on the number of requests $N(\sigma_1) = n$, for some $n \in \mathbb{N}$, the arrival times in $\sigma_1$ follow the same distribution as the order statistics corresponding to $n$ independent random variables uniformly distributed over $[0, p]$. 
Let us denote these variables by $A_1, A_2, \ldots, A_n$. 
Consider any request $r_j \in \sigma_1$ that arrived at time $t(r_j) = A_j$ and is still pending at time $p$.
It is easy to notice that the variable $U_j$ representing the delay cost $r_j$ incurred until $p$ also follows a uniform distribution $[0, p]$ as it holds that $U_j = p - A_j$, i.e., $U_j \sim \mathcal{U}(p)$. 
Thus, when we condition on $N(\sigma_i) = n$, we can write $D(\sigma_1) = \sum_{j=1}^n U_j$ as the sum of $n$ uniform variables representing the waiting times. 
This allows us to rewrite the right-hand side of \eqref{eq:edge_opt_1} as
\begin{equation} \label{eq:edge_opt_2} 
    \E\big[\min(w, D(\sigma_1))\big] = \sum_{n=1}^{\infty} \P\big[N(\sigma_1) = n\big] \cdot \E\left[\min\Bigg(w,\sum_{j=1}^n U_j\Bigg)\ \Bigg|\ N(\sigma_i) = n\right], 
\end{equation}
where the expectation on the right side is taken over all sequences $\{U_j: j \in \mathbb{N}\}$ of independent uniform variables in $[0,p]$. We now estimate these expectations. 

\begin{claim}
    Given $p, w > 0$, an integer $n \ge 1$, and a sequence $\{U_j: j \in \mathbb{N}\}$ of independent uniform random variables defined over $[0,p]$, if it holds that $n p \ge w$, then
    \begin{equation*}
        \E\left[\min\Bigg(w,\sum_{j=1}^n U_j\Bigg)\right] \ge w \left(1-\frac{w}{2n p}\right).
    \end{equation*}
\end{claim}

\begin{proof}
    When $n = 1$, we have 
    \begin{equation} \label{eq:bounding_at_least_one_req}
    \begin{split}
        \E\big[\min(w, U_1)\big] & = w \big(1 - \P\big[U_1 \le w\big]\big) + \P\big[U_1 \le w\big] \cdot \E\big[U_1 \mid U_1 \le w\big] \\
        & = w \left(1 - \frac{w}{p}\right) + \frac{w}{p} \cdot \frac{w}{2} = w \left(1 - \frac{w}{2p}\right).
    \end{split}
    \end{equation}
    When $n \ge 2$, we notice that $\min(w, \sum_{j=1}^n U_j) \ge \sum_{j \in [n]} \min(w / n, U_j)$.
    Indeed, let us denote the variables on the right-hand side by $B_j$, i.e., $B_j := \min(w/n, U_j)$ for $j \in [n]$. Whenever the sum of $U_j$ realizations is smaller than $w$, the sum of $B_j$ values cannot be larger as each of them is upper bounded by $U_j$. 
    On the other hand, in case $U_j$ sum up to something bigger than $w$, the sum of $B_j$s is not larger than $w$ as there are $n$ of them, each upper bounded by $w/n$. 
    Thus, we can use the monotonicity of expectation. Moreover, since $w / n \le p$, we can follow the steps of \eqref{eq:bounding_at_least_one_req} to get
    \begin{equation*}
        \E\left[\min\Bigg(w, \sum_{j=1}^n U_j\Bigg)\right] \geq \sum_{j \in [n]}\E\bigg[\min\left(\frac{w}{n}, U_j\right)\bigg] \ge \sum_{j\in[n]} \frac{w}{n} \left(1-\frac{w}{2n p}\right) = w \left(1-\frac{w}{2n p}\right).
    \end{equation*}
    As a result, we proved this claim.
\end{proof}

Let $n_0\in \mathbb{N}$. 
If $w \le n_0 p$, then by applying the bound of the claim in equation \eqref{eq:edge_opt_2}, we obtain
\begin{align} \label{eq:edge_opt_3}
    \E\big[\min(w, D(\sigma_1))\big] & \geq \sum_{n=n_0}^{\infty} \P\big[N(\sigma_1) = n\big] \cdot w \left(1-\frac{w}{2n p}\right) \geq \sum_{n=n_0}^{\infty} \P\big[N(\sigma_1) = n\big] \cdot w \left(1-\frac{w}{2n_0 p}\right) \nonumber \\
    & = \P\big[N(\sigma_1) \geq n_0\big] \cdot w \left(1 - \frac{w}{2n_0 p}\right).
\end{align}

In order to obtained the desired bound on the expected cost of the optimal schedule, we now define the suitable values of $p$ and $n_0$ depending on whether $w \lambda \le 1$ or $w \lambda > 1$.

\paragraph{Case $w \lambda \le 1$.} We define $n_0 = 1$ and $p = 1 / \lambda$. 
Then, $n_0 p = 1 / \lambda \geq w$. 
Moreover, $\P(N(\sigma_1) \ge 1)$ equals $1 - e^{-p\lambda} = 1 - e^{-1}$. 
Thus, combining \eqref{eq:edge_opt_1}, \eqref{eq:edge_opt_3} and the values of $n_0$ and $p$, we obtain
\begin{equation} \label{eq:case1_edge}
    \E\big[\cost(\opt(\sigma))\big] \geq \frac{\tau}{p} \left(1-e^{p\lambda}\right)  w  \left(1 - \frac{w}{2p}\right) = \tau \lambda w \Big(1-e^{-1}\Big)\left(1-\frac{w\lambda}{2}\right) \geq \frac{1-e^{-1}}{2} \tau \lambda w.
\end{equation}

\paragraph{Case $w \lambda > 1$.} 
 We define $p = \sqrt{2w/\lambda}$ and $n_0 = \lceil\lambda p\rceil$. 
 We have $n_0 p \ge \lambda p^2 = 2w > w$. 
 Moreover, by Proposition \ref{prop:poisson_expected_number_requests}, we get that $\P(N(\sigma_1) \ge n_0) \ge 1 / 2$. 
 Thus, combining \eqref{eq:edge_opt_1}, \eqref{eq:edge_opt_3} and the value of $n_0$ and $p$, we obtain
\begin{equation*}
    \E\big[\cost(\opt(\sigma))\big] \geq \frac{\tau}{p}  \P\big[N(\sigma_1) \ge n_0\big] w \left(1-\frac{w}{2n_0 p}\right) \geq \frac{\tau}{2} \sqrt{\frac{\lambda}{2w}} w \left(1-\frac{w}{2\cdot 2w}\right) = \frac{3}{8\sqrt{2}} \tau \sqrt{w\lambda}.
\end{equation*}
This concludes the proof of Lemma \ref{lemma:opt_edge}.

\section{Overview}
\label{section:overview}
We now give an overview of the following sections. 
Inspired by the two strategies for the single edge instance, we define two types of stochastic instances: the \emph{light} instances for which the strategy of serving requests instantly achieves a constant RoE, and the \emph{heavy} instances for which the strategy of serving requests periodically achieves a constant RoE. Heavy and light instances are defined precisely below (Definitions \ref{definition:light_instance} and \ref{definition:heavy_instance}) and generalizes the notions of heavy and light single-edge trees studied in the previous section. 

We first define the light instances by extending the notion of \emph{heaviness} for an arbitrary tree. 

\begin{restatable}{definition}{definitionlight} \label{definition:light_instance}
    An instance $(T, \llambda)$ is called light if $\pi(T, \llambda) \le 1$, where $\pi(T, \llambda)$ is 
    \begin{equation*}
        \pi(T, \llambda) := \sum_{u \in V(T)} \llambda(u) \cdot d(u, \gamma(T)),
    \end{equation*}
    is called the \emph{heaviness} of the instance. 
\end{restatable}

We show in Section \ref{section:lightcase} that for a light instance, serving the requests immediately at the their arrival time achieves a constant RoE. We refer to the schedule produced with this strategy (see Algorithm \ref{pseudocode:instant}) on a sequence of requests $\sigma$ by $\instant(\sigma)$. 

\RestyleAlgo{boxruled}
\LinesNumbered
\SetAlgoVlined
\begin{algorithm}
\caption{INSTANT}
\label{pseudocode:instant}
\KwIn{A sequence $\sigma$ of requests.}
\KwOut{A valid schedule $\instant(\sigma)$ for $\sigma$.}
\For{each request $r\in \sigma$ at its arrival time $t(r)$}{
    Create a service $(t(r),\{r\})$. 
    }
\end{algorithm}
Notice that this algorithm does not require the knowledge of the arrival rates. 

\begin{restatable}{theorem}{theoreminstantlight} \label{theorem:instant_light}
    INSTANT has a ratio of expectations of $16 / (3 - 3e^{-1}) < 8.44$ for light instances. 
\end{restatable}

We prove this theorem in Section \ref{section:lightcase}. We now turn our attention to heavy instances. An instance $(T, \llambda)$ is heavy if for every subtree $T'\subseteq T$, we have $\pi(T',\llambda)>1$. By monotonicity of $\pi(\cdot, \llambda)$, we obtain the following equivalent definition.  Recall that for a vertex $u\in V(U)$, $w_u$ denotes the weight of the edge incident to $u$ on the path from $\gamma(T)$ to $u$.   

\begin{restatable}{definition}{definitionheavy} \label{definition:heavy_instance}
    An instance $(T, \llambda)$ is called \emph{heavy} if $w_u \ge 1 / \llambda(u)$ for all $u \in V(T)$ with $\llambda(u) > 0$. 
\end{restatable}

To give some intuition, suppose that $u$ is a vertex of an heavy instance, and $r$ and $r'$ are two consecutive (random) requests located on $u$. Then, the expected duration between their arrival times is $1/\llambda(u)< w_u$. This suggests that to minimize the cost, we should in average gather $r$ and $r'$ into the same service, is order to avoid to pay twice the weight cost $w_u$. Since  we expect services to serve a group of two or more requests, our stochastic assumptions suggests to the services must follow some form of regularity. 

In Section \ref{section:heavycase}, we present an algorithm called \plan, that given an heavy instance $(T,\llambda)$, computes for each vertex $u\in V(T)$ a period $p_u > 0$, and will serve $u$ at every time that is a multiple of $p_u$. One intuitive property of these periods $\{p_u: u \in V(T)\}$ is that the longer the distance to the root, the longer the period. While losing only a constant fraction of the expected cost, we choose the periods to be powers of $2$. This enables us to optimize the weights of the services on the long run. One interesting feature of our algorithm is that it acts ``blindly'': the algorithm does not need to know the requests, but only the arrival rate of each point! Indeed, our algorithm may serve a vertex where there are no pending requests. 
For the detail of the PLAN algorithm, see Section \ref{section:heavycase}. 

\begin{restatable}{theorem}{theoremplanheavy} \label{theorem:plan_heavy}
    PLAN has a ratio of expectations $64/3 < 21.34$ for heavy instances. 
\end{restatable}

We remark that light instances and heavy instances are not complementary: there are instances that are neither light nor heavy.\footnote{Furthermore, there exists a stochastic MLA instance where the ratio of expectations are both unbounded if INSTANT or PLAN is directed applied to deal with. See Appendix \ref{appendix:overview} for details. } 
In Section \ref{section:general}, we focus on the general case of arbitrary instances. The strategy there is to partition the tree (and the sequence of requests) into two groups of vertices (two groups of requests), so that the first group corresponds to a light instance where we can apply the instant strategy while the second group corresponds to a heavy instance where we can apply a periodic strategy. However, this correspondence for the heavy group is not straightforward. 
For this, we need to define an \emph{augmented tree} that is a copy of the original tree, with the addition of some carefully chosen vertices. Each new vertex is associated with a subset of vertices of the original tree called \emph{part}. We then define an arrival rate for each of these new vertices that is equal to the sum of the arrival rates of the vertices in the corresponding part. We show that this defines an heavy instance on which we can apply the algorithm \plan. For each service made by $\plan$  on each of this new vertices, we serve all the pending requests in the corresponding part. The full description of this algorithm, called \gen, is given in Section \ref{section:general}. We show that this algorithm achieves a constant ratio of expectations. 

\begin{restatable}{theorem}{theoremgeneralany} \label{theorem:general_any}
    GEN has a ratio of expectations of 210 for an arbitrary stochastic instance. 
\end{restatable}

\section{Light instances}
\label{section:lightcase}
In this section, we prove Theorem \ref{theorem:instant_light} that we recall below. 

\theoreminstantlight*

Recall that an instance $(T,\llambda)$ is light if $\pi(T,\llambda) := \sum_{u \in V(T)} (\llambda(u) \cdot d(u,\gamma(T)))\le 1$. When the context is clear we simply write $\pi(T) = \pi(T, \llambda)$. The proof of the theorem easily follows from the two following lemmas, that respectively estimates the expected cost of the algorithm, and give a lower bound on the expected cost of the optimal offline schedule. 

\begin{lemma} \label{lemma:instant_light_upper-bound}
    Let $(T,\llambda)$ be a light instance, and $\tau>0$. Then, 
    \begin{equation*}
        \E\left[\cost(\instant(\sigma))\mid \sigma\sim(T,\llambda)^\tau\right] =  \tau \cdot \pi(T,\llambda).
    \end{equation*}
\end{lemma}

\begin{lemma} \label{lemma:opt_light_lower-bound}
    Let $(T,\llambda)$ be a light instance, and $\tau>0$. Then,
    \begin{equation*}
        \E\left[\cost(\opt(\sigma)) \mid \sigma\sim(T,\llambda)^\tau\right] \ge \frac{3}{16}(1 - e^{-1}) \cdot \tau \cdot \pi(T,\llambda). 
    \end{equation*}
\end{lemma}

\begin{proof}[Proof of Lemma \ref{lemma:instant_light_upper-bound}.]
Let $\sigma$ be a sequence of requests for $T$ of duration $\tau$ and let $u\in V(T)$. For each request located on $u$, the algorithm issues a service of cost $d(u,\gamma(T))$ (notice that the delay cost is equal to zero). 
Let $N(\sigma|_u)$ denotes the number of requests in $\sigma$ that are located at $u$. By Proposition \ref{prop:poisson_expected_number_requests}, we know that $\E[N(\sigma|_u)]=\tau \llambda(u)$. Thus, we have  
\begin{align} \label{eq:light_instant}
    \E[\cost(\instant(\sigma))] &= \E\Big[\sum_{u\in V(T)}N(\sigma|_u) \cdot d(u,\gamma(T))\Big] = \sum_{u \in V(T)} \tau \cdot \llambda(u) \cdot d(u,\gamma(T))\nonumber \\
    &=\tau \cdot \pi(T,\llambda),
\end{align}
which concludes the proof. 
\end{proof}

\begin{proof}[Proof of Lemma \ref{lemma:opt_light_lower-bound}.] 
Without loss of generality, we assume that $\tau$ is a multiple of $1 / \lambda$. The plan of the proof is to associate to $(T,\llambda)$ a specific family of single-edge light instances. We then apply the bounds proved in Section \ref{section:warmup} to establish the bound of the lemma. 

For each integer $j \in \mathbb{Z}$, we denote 
\begin{equation*}
    V_j := \{u \in V(T): 2^j \le d_T(u,\gamma(T)) < 2^{j+1}\}.
\end{equation*}
Let $\sigma$ be a sequence of requests for $T$.  We denote  
\begin{equation*}
    \sigma_j := \sigma|_{V_j}=\{r \in \sigma: \ell(r)\in V_j\}.
\end{equation*}

For each $j\in \mathbb{Z}$, we create a single-edge stochastic instance as follows. Let $e_j$ denote a single-edge tree of weight $w_j=2^{j-1}$. Let $\lambda_j:=\llambda(V_j)$ denote the arrival rate of the vertex in $e_j$ (that is not the root).  
We construct $\sigma'_j$ as a sequence of requests for $e_j$ with the same arrival times as requests in $\sigma_j$. 
Let $S(\sigma)$ denote a schedule for $\sigma$.
We create a schedule $S(\sigma'_j)$ for $\sigma'_j$ as follows: for each service $s \in S(\sigma)$ that serves at least one request from $V_j$, add a service $s'_j$ in $S(\sigma'_j)$ with the same service time to serve the corresponding requests in $\sigma'_j$. It is clear from the construction that if $S(\sigma)$ is a valid schedule for $\sigma$, then for each $j\in \mathbb{Z}$, $S(\sigma'_j)$ is a valid schedule for $\sigma'_j$. Further, we have the following inequality on the cost of these schedules.  
\begin{equation} \label{eq:light-lowerboundingscheme}
    \cost(S(\sigma),T) \ge \sum_{j\in \mathbb{Z}} \cost(S(\sigma'_j),e_j). 
\end{equation}

Indeed, first notice that the delay costs are the same, i.e., $\delay(S(\sigma)) = \sum_j \delay(S(\sigma'_j))$. 
We now focus on the weight of these schedules. 
Let $s$ be a service in $S(\sigma)$ which serves a subset of requests $R(s) \subseteq \sigma$ at time $t$. 
Let $k \in \mathbb{Z}$ be the largest index such that $R(s) \cap \sigma_k \neq \emptyset$. 
This means that $s$ serves a request that is at distance at least $2^k$ from the root, and then 
\begin{equation*}
    \weight(s) \ge 2^k \ge \sum_{j \le k} 2^{j-1} \ge \sum_{j \le k} \weight(s'_j)=\sum_{j \in\mathbb{Z}} \weight(s'_j),
\end{equation*}
where $s'_j$ denotes the service created at time $t$ that serves the requests in $\sigma'_j$ that correspond to $R(s)\cap \sigma_j$. 
As a result, 
\begin{equation*}
    \cost(S(\sigma),T) \ge \sum_{j\in\mathbb{Z}} \cost(S(\sigma'_j),e_j) \ge \sum_{j\in\mathbb{Z}} \cost(\opt(\sigma'_j), e_j),
\end{equation*}
where $\opt(\sigma'_j)$ denotes the optimal schedule for $\sigma'_j$. 
Since this holds for any valid schedule $S(\sigma)$, we obtain
\begin{equation*}
    \cost(\opt(\sigma),T) \ge \sum_j \cost(\opt(\sigma'_j),e_j).
\end{equation*}

Due to the equivalence between the distributed and centralized Poisson arrival model, we know that $\sigma \sim (T,\llambda)^\tau$ implies $\sigma_j' \sim (e_j,\lambda_j)^\tau$.
By taking expectation over all the random sequences $\sigma \sim (T,\llambda)^\tau$, we have
\begin{equation*}
    \E[\cost(\opt(\sigma),T)] \ge \sum_{j\in\mathbb{Z}} \E\left[\cost(\opt(\sigma'_j),e_j)\mid \sigma_j' \sim (e_j,\lambda_j)^\tau\right].
\end{equation*}

We now show that for each $i\in \mathbb{Z}$, $(e_j,\lambda_j)$ is a light instance. 
Indeed, we have
\begin{align*}
    \pi(e_j,\lambda_j)
    &= w_j \cdot \lambda_j= 2^{j-1} \cdot \lambda_j = 2^{j-1} \cdot \sum_{u \in V_j} \llambda(u) = \frac{1}{2} \sum_{u \in V_j} \llambda(u) \cdot 2^j \le \frac{1}{2} \sum_{u \in V_j} \llambda(u) \cdot d(u,\gamma(T)) \\
    &\le \frac{1}{2} \sum_{u \in V(T)} \llambda(u) \cdot d(u,\gamma(T)) = \frac{1}{2} \cdot \pi(T,\llambda) \le \frac{1}{2}, 
\end{align*}
since $(T,\llambda)$ is light. 
Thus, for each $i\in \mathbb{Z}$, by plugging $p=1/\lambda_j$ in equation \eqref{eq:case1_edge} in the proof of Lemma \ref{lemma:opt_edge}, we obtain
    \begin{align*}
     \E\left[\cost(\opt(\sigma'_j),e_j)\mid \sigma_j' \sim (e_j,\lambda_j)^\tau\right] &\ge \frac{\tau}{p}  (1-e^{p\lambda})  w_j  (1 - \frac{w_j}{2p})\ge   (1 - \frac{w_j\lambda_j}{2}) (1-e^{-1})  {\tau}{w_j\lambda_j}  \\
     &\ge \frac{3}{4} (1-e^{-1}) {\tau}{w_j\lambda_j}.
\end{align*}
On the other hand, we have
\begin{equation*}
    w_j \cdot \lambda_j = 2^{j-1}\sum_{u \in V_j} \llambda(u) = \frac{1}{4}\sum_{u \in V_j} 2^{j+1} \llambda(u) \ge \frac{1}{4} \sum_{u \in V_j} d(u,\gamma(T)) \cdot \llambda(u). 
\end{equation*}
Putting everything together we obtain the desired bound:
\begin{align*}
    \E[\cost(\opt(\sigma))] 
    &\ge \sum_{j\in\mathbb{Z}} \E\left[\cost(\opt(\sigma'_j),e_j)\mid \sigma_j' \sim (e_j,\lambda_j)^\tau\right] 
    \ge \sum_{j\in\mathbb{Z}}  \frac{3}{4} (1-e^{-1})  {\tau}{w_j\lambda_j} \\
    &
    \ge \frac{3}{4} (1-e^{-1}) \cdot \frac{1}{4} \cdot \tau \cdot \sum_j \sum_{u \in V_j} d(u,\gamma(T)) \cdot \llambda(u) \\
    &= \frac{3}{16} (1-e^{-1}) \cdot \tau \cdot \sum_{u \in V(T)} d(u,\gamma(T)) \cdot \llambda(u) 
    = \frac{3}{16} (1-e^{-1}) \cdot \tau \cdot \pi(T, \llambda).
\end{align*}
This concludes the proof. 
\end{proof}

\section{Heavy instances}
\label{section:heavycase}

In this section, we focus on heavy MLA instances. 
Let us first recall their definition.
\definitionheavy*

\noindent
To solve the problem, when restricted to this class of MLA instances, we define a new algorithm $\plan$.
In the main theorem of this section, we prove that within this class, $\plan$ achieves a constant ratio of expectations. 

Our approach can be seen as a generalization of the algorithm for a single-edge case. 
Once again, we serve the requests periodically, although this time, we may assign different periods for different vertices. 
Intuitively, vertices closer to the root and having a greater arrival rate should be served more frequently. 
For this reason, the $\plan$ algorithm generates a partition $P$ of a given tree $T$ into a family of subtrees (clusters) and assigns a specific period to each of them. 

The partition procedure allows us to analyze each cluster separately. 
Thus, we can assume that from now on, we are restricted to a given subtree $T' \in P$.
To lower bound the cost generated by $\opt$ on $T'$, we split the weight of $T'$ among its vertices using a saturation procedure. 
This action allows us to say that for each vertex $v$, the optimal algorithm either covers the delay cost of all the requests arriving at $v$ within a given time horizon or pays some share of the service cost. 
The last step is to round the periods assigned to the subtrees in $P$ to minimize the cost of $\plan$.
In what follows, we present the details of our approach.

\subsection{Periodical algorithm PLAN}

As mentioned before, the main idea is to split tree $T$ rooted at vertex $\gamma$ into a family of subtrees and serve each of them periodically. 
In other words, we aim to find a partition $P = \{T_1, T_2, \ldots, T_k\}$ of $T$ where each subtree $T_i$ besides the one containing $\gamma$ is rooted at the leaf vertex of another subtree.
At the same time, we assign each subtree $T_i$ some period $p_i$.
To decide how to choose the values of $p_i$s, let us recall how we picked the period for a single-edge case. 
In that setting, for the period $p$ we had an equality between the expected delay cost $\lambda / 2 \cdot p^2$ at the leaf $u$ and the weight $w$ of the edge. 
Thus, the intuition behind the $\plan$ algorithm is as follows.

We start by assigning each vertex $v \in T$ a process that saturates the edge connecting it to the parent at the pace of $\llambda(v) / 2 \cdot t^2$, i.e., within the time interval $[t, t + \epsilon]$ it saturates the weight of $\llambda(v) / 2 \cdot ((t + \epsilon)^2 - t^2)$. 
Whenever an edge gets saturated, the processes that contributed to this outcome join forces with the processes that are still saturating the closest ascendant edge. 
As the saturation procedure within the whole tree $T$ reaches the root $\gamma$, we cluster all the vertices corresponding to the processes that made it possible into the first subtree $T_1$. 
Moreover, we set the period of $T_1$ to the time it got saturated. 
After this action, we are left with a partially saturated forest having the leaves of $T_1$ as the root nodes. 
The procedure, however, follows the same rules, splitting the forest further into subtrees $T_2, \ldots, T_k$. 

To simplify the formal description of our algorithm, we first introduce some new notations. 
Let $p(v)$ denote the saturation process defined for a given vertex $v$. 
As mentioned before, we define it to saturate the parent edge at the pace of $\llambda(v) / 2 \cdot t^2$. 
Moreover, we extend this notation to the subsets of vertices, i.e., we say that $p(S)$ is the saturation process where all the vertices in $S$ cooperate to cover the cost of an edge. 
The pace this time is equal to $\llambda(S) / 2 \cdot t^2$. 
To trace which vertices cooperate at a given moment and which edge they saturate, we denote the subset of vertices that $v$ works with by $S(v)$ and the edge they saturate by $e(v)$. 
We also define a method $\texttt{join}(u ,v)$ that takes as the arguments two vertices and joins the subsets they belong to. 
It can be called only when the saturation process of $S(u)$ reaches $v$. 
Formally, at this moment the \texttt{join} method merges subset $S(u)$ with $S(v)$ and sets $e(v)$ as the outcome of the function $e$ on all the vertices in the new set. 
It also updates the saturation pace of the new set. 
Now, we present the pseudo-code for $\plan$ as Algorithm \ref{alg:plan}.
For a visual support to illustrate the saturating process, see Figure \ref{fig:plan_part_one}.

\RestyleAlgo{boxruled}
\LinesNumbered
\SetAlgoVlined
\begin{algorithm}
\caption{PLAN (part I)}
\label{alg:plan}
\KwIn{an heavy instance $(T, \llambda)$ with tree $T$ rooted at $\gamma$}
\KwOut{a partition $P = \{T_1, T_2, \ldots, T_k\}$ of $T$ with each subtree $T_i$ assigned a period $p_i$}
let $R$ be the set of roots, initially $R = \{\gamma\}$ \hfill $\vartriangleright$ \textit{at the beginning we work with the whole tree} \\
\For{each vertex $v \in V(T)$}{
    define the saturation process $p(v)$ as described before \\
    set $S(v) := \{v\}$ and $e(v) := \parent{v}$
}
start the clock at time $0$ \\   \label{line:start_the_clock}
\While{there exist some unclustered vertices in $T$}{
    wait until the first time $t_e$ when an edge $e = (u, v)$ gets saturated \\
    \If{$v \not\in R$}{
        $\texttt{join}(u, v)$   \label{line:join_subsets}
    }
    \Else{
        add cluster $C := S(u) \cup \{v\}$ to partition $P$ \hfill $\vartriangleright$ \textit{we reached the root vertex $v$} \\
        set the period $p$ for $C$ to be equal to $t_e$ \\
        set the saturation pace for $C$ to 0 \\
        extend $R$ by the leaves in $C$ \hfill $\vartriangleright$ \textit{leaves of $C$ are the roots of the remaining forest}
    }
}
\end{algorithm}

\begin{figure*}
    \centering
    \includegraphics[width=0.9\textwidth]{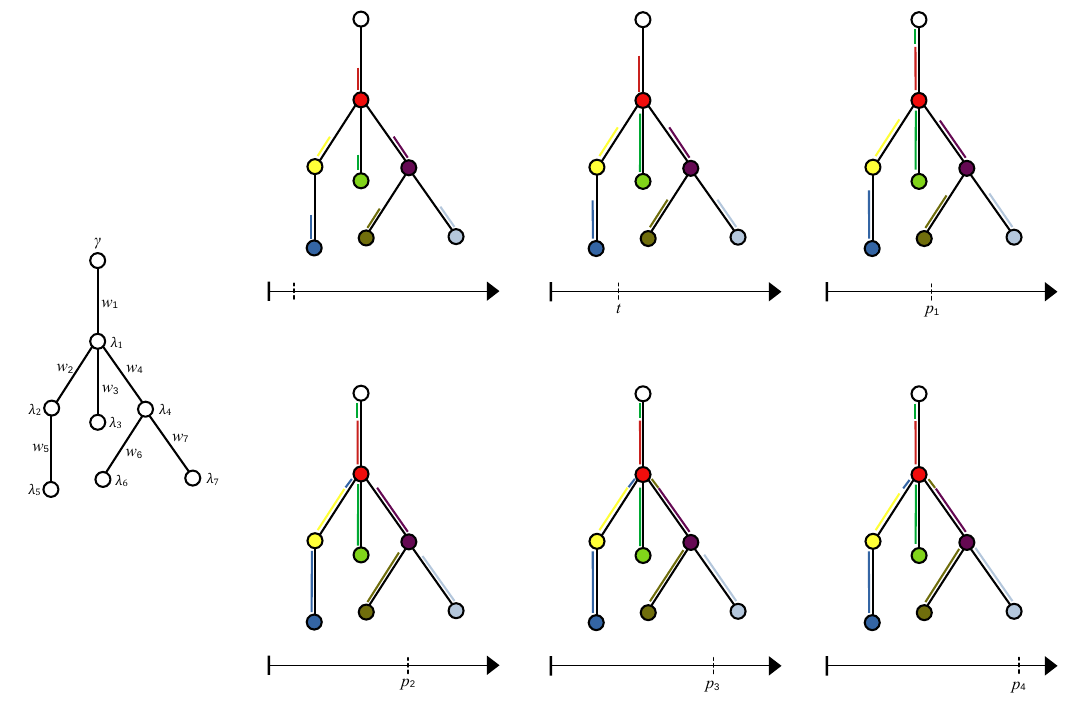}
    \caption{Here is an example to show how Algorithm \ref{alg:plan} works on an heavy instance. Given the tree consisting of 7 vertices (with $w_i \ge 1/ \lambda_i$ for each vertex $i \in [7]$ marked in different color), we use the length of the colored line to denote the saturated amount (i.e., $\llambda_i / 2 \cdot t^2$) of a vertex $i$ at any time $t$. 
    At time $p_1$, the subtree $T_1$ including vertices 1 and 3 is determined; similarly, $T_2$ includes vertices 2 and 5 at time $t_2$; $T_3$ includes vertices 4 and 6 at time $p_3$; and $T_4$ includes vertex 7 at time $p_4$.}
    \label{fig:plan_part_one}
\end{figure*}

We start by listing some properties of the partition generated by this algorithm.

\begin{proposition} \label{prop:heavy_plan_output_properties}
    Let $(T, \gamma, \llambda, \textbf{w})$ be a heavy instance and let $P = \{T_1, T_2, \ldots, T_k\}$ be the partition generated on it by Algorithm \ref{alg:plan}. We denote the period corresponding to $T_i$ by $p_i$. Assuming that $T_i$s are listed in the order they were added to $P$, it holds that:
    \begin{enumerate}
        \item each $T_i$ is a rooted subtree of $T$;
        \item the periods are increasing, i.e., $1 \leq p_1 \leq p_2 \leq \ldots \leq p_k$;
        \item each vertex $v \in T_i$ saturated exactly $\llambda(v) / 2 \cdot p_i^2$ along the path to the root of $T_i$.
    \end{enumerate}
\end{proposition}

\begin{proof}
    To show that the first property is satisfied, we proceed by induction. Initially, we have that each subset $S(v)$ for $v \in T$ is a single vertex and thus it forms a subtree. Then, in line \ref{line:join_subsets}, we can merge two subtrees only if an edge connects them. Thus, the \texttt{join} call also creates a new subset that induces a subtree. Finally, we notice that we cluster a subset only as it reaches a vertex from the set $R$. It becomes the root of this subtree, which implies the desired property.

    The second property follows straight from the assumption that we started the clock at time 0 in line \ref{line:start_the_clock} and we process the edges in order they get saturated, i.e., there is no going back in time. Similarly, the last property is implied by the definition of the saturation process.
\end{proof}

\subsection{Lower bounding OPT}

In this subsection, we lower bound the total cost incurred by the optimal offline schedule $\opt$ on a heavy instance. Let us first consider each subtree $T_i$ generated by Algorithm \ref{alg:plan} separately.

\begin{lemma}
    Let $(T, \llambda)$ be a heavy instance. We denote the partition generated for it by Algorithm \ref{alg:plan} by $P = \{T_1, T_2, \ldots, T_k\}$ and the period corresponding to $T_i$ by $p_i$ for all $i \in [k]$. Let $T_i$ be any subtree in $P$, and let us define $\sigma_i$ as a random sequence of requests arriving within the MLA instance restricted to $T_i$ over a time horizon $\tau$. We assume that $\tau$ is a multiple of $p_i$. It holds that
    \begin{equation*}
        \E[\cost(\opt(\sigma_i), T_i) \mid \sigma_i \sim (T_i, \llambda|_{T_i})^\tau] \geq \frac{3}{16} \cdot w(T_i) \cdot \frac{\tau}{p_i}.
    \end{equation*}
\end{lemma}
\begin{proof}
    We use the same approach as in Section \ref{section:warmup} and first focus on lower bounding the cost incurred by $\opt$ within a shorter time interval --- for now, we set the horizon to $p_i$. 
    By Proposition \ref{prop:heavy_plan_output_properties}, we have that by time $p_i$ each vertex $v \in T_i$ saturates the weight of $\hat{w}_v := \llambda(v) / 2 \cdot p_i^2$ along the path to the root. 
    Thus, whenever $\opt$ issues a service that contains $v$, we can distribute the service cost among the served vertices and say that $v$ needs to cover $\hat{w}_v$ share.  

    By the definition of $\plan$, we know that the sum of $\hat{w}_v$ over all the vertices $v \in T_i$ is equal to the weight $w(T_i)$, as $T_i$ is fully saturated at moment $p_i$. Moreover, by the definition of a heavy instance, we have that $w_v \geq 1/\llambda(v)$ for each $v \in T_i$. Combining it all together gives us
    \begin{equation} \label{ineq:heavy_total_saturated_cost}
        \sum_{v \in T_i} \frac{\llambda(v) \cdot p_i^2}{2} = \sum_{v \in T_i} \hat{w}_v = \sum_{v \in T_i} w_v \geq \sum_{v \in T_i} \frac{1}{\llambda(v)}.
    \end{equation}

    Now, we apply the single-edge case analysis for some of the vertices in $T_i$. 
    To be more precise, we focus on the case where the product of the arrival rate and the weight of a given vertex was at least 1. 
    The crucial assumption there was to guarantee that the parameter for the Poisson arrival variable, i.e., the product of the arrival rate and the period, was at least 1. 
    Thus, in the current scenario, we need to check for which vertices $v \in T_i$ it holds that $\llambda(v)p_i \geq 1$. 

    Here, we use a different approach and upper bound the contribution to the total saturated cost $\sum_{v \in T_i} \hat{w}_v$ incurred by the vertices that do not satisfy this property. 
    Let us denote the set of such vertices by $L_i$. 
    We have that $p_i \leq 1/\llambda(v)$ for each $v \in L_i$. 
    Thus, combining this with inequality \eqref{ineq:heavy_total_saturated_cost}, implies that
    \begin{equation}
        \sum_{v \in L_i} \hat{w}_v = \sum_{v \in L_i} \frac{\llambda(v) \cdot p_i^2}{2} \leq \frac{1}{2} \sum_{v \in L_i} \frac{\llambda(v)}{\llambda^2(v)} = \frac{1}{2} \sum_{v \in L_i} \frac{1}{\llambda(v)} \leq \frac{1}{2} \sum_{v \in T_i} \frac{1}{\llambda(v)} \leq \frac{1}{2} \sum_{v \in T_i} w_v,
    \end{equation}
    which proves that at least half of the saturation cost comes from the heavy vertices.

    As we apply the single-edge case analysis for all the vertices in $T_i \setminus L_i$, saying that within each period $p_i$, $\opt$ has to pay either the service cost of $\hat{w}_v$ or the total delay cost generated at vertex $v$ for each $v \in T_i$, we obtain
    \begin{equation*}
        \E[\cost(\opt(\sigma_i))] \geq \sum_{v \in T_i \setminus L_i} \frac{3}{8} \cdot \hat{w}_v \cdot \frac{\tau}{p_i} = \frac{3}{8} \sum_{v \in T_i \setminus L_i} \hat{w}_v \cdot \frac{\tau}{p_i} \geq \frac{3}{8} \cdot \frac{1}{2} \sum_{v \in T_i} w_v \cdot \frac{\tau}{p_i} = \frac{3}{16} \cdot w(T_i) \cdot \frac{\tau}{p_i},
    \end{equation*}
    which ends the proof.
\end{proof}

\subsection{Cost analysis for PLAN}

Let us start by assuming that we serve all the subtrees $T_i$ generated by Algorithm \ref{alg:plan} periodically according to the periods $p_i$. 
In this setting, to serve any cluster besides the one containing the root vertex $\gamma$, not only we need to cover the service cost of the cluster vertices but also the cost of the path connecting them to $\gamma$. 
Since we only know how to lower bound the cost incurred by $\opt$  on the clusters, we improve the $\plan$ algorithm to get rid of this problem.
The idea is to round the periods $p_i$s to be of form $2^{e_i} p_1$ for some positive integers $e_i$. 
Thus, whenever we need to serve some cluster $S_i$, we know that we get to serve all the clusters generated before it as well.
Formally, our approach is presented in Algorithm \ref{alg:plan_2}.
\RestyleAlgo{boxruled}
\LinesNumbered
\SetAlgoVlined
\begin{algorithm}
\caption{PLAN (part II)}
\label{alg:plan_2}
\KwIn{increasing sequence of periods $(p_1, p_2, \ldots, p_k)$ obtained from Algorithm \ref{alg:plan}}
\KwOut{new sequence of periods $(\hat{p}_1, \hat{p}_2, \ldots, \hat{p}_k)$}
\For{$i \in \{2, 3, \ldots, k\}$}{
    find $e_i \in \mathbb{Z}_+$ such that $2^{e_i} p_1 \leq p_i < 2^{e_i + 1} p_1$ \\
    set $\hat{p}_i := 2^{e_i} p_1$
}
\end{algorithm}

\noindent Finally, we define the algorithm $\plan$ to serve the requests periodically, according to the new periods (see Algorithm \ref{alg:plan_3}).
\RestyleAlgo{boxruled}
\LinesNumbered
\SetAlgoVlined
\begin{algorithm}
\caption{PLAN (part III)}
\label{alg:plan_3}
\KwIn{heavy instance $(T, \llambda)$ and a sequence of requests $\sigma$ in $[0, \tau]$} 
\KwOut{valid schedule $\plan(\sigma)$ for $\sigma$}
let $P = \{T_1, T_2, \ldots, T_k\}$ be the partition generated by Algorithm \ref{alg:plan}; \\
let $(\hat{p}_1, \hat{p}_2, \ldots, \hat{p}_k)$ be a sequence of periods obtained from Algorithm \ref{alg:plan_2}; \\
\For{$i \in [k]$}{
    \For{$j \in [\lfloor \tau / \hat{p}_i \rfloor]$}{
        schedule a service that serves $T_i$ at time $j \hat{p}_i$
    }
    \If{$\tau$ is not a multiple of $\hat{p}_i$}{
        schedule a service that serves $T_i$ at time $\tau$   \label{line:plan_last_extra_service}
    }
}
\end{algorithm}

To conclude this section, we prove Theorem \ref{theorem:plan_heavy} (restated below).
\theoremplanheavy*
\begin{proof}
    Let $(T, \llambda)$ be a heavy instance with tree $T$ rooted at $\gamma$ and let $P = \{T_1, T_2, \ldots, T_k\}$ be the partition generated for it by Algorithm \ref{alg:plan}. Moreover, let $(p_1, p_2, \ldots, p_k)$ and $(\hat{p}_1, \hat{p}_2, \ldots, \hat{p}_k)$ denote the periods obtained from Algorithm \ref{alg:plan} and Algorithm \ref{alg:plan_2}, respectively. Here, we analyze the cost generated by $\plan$ (Algorithm \ref{alg:plan_3}) on a random sequence of requests $\sigma \sim (T, \llambda)^\tau$, where the time horizon $\tau$ is a multiple of $2\hat{p}_k$.

    Notice that since we align the periods to be of form $2^l\hat{p}_1$ for some positive integer $l$, whenever $\plan$ serves some tree $T_i$, it serves all the trees containing the path from $T_i$ to $\gamma$ at the same time. 
    Thus, the service cost can be estimated on the subtree level. 
    Moreover, since for each $i \in [k]$ it holds that period $\hat{p}_i \leq p_i$, the expected delay cost incurred within $[0, \hat{p}_i]$ does not exceed $w(T_i)$. 
    Thus, denoting $\sigma_i \sim (T_i, \llambda|_{T_i})^{\hat{p}_i}$ for $i \in [k]$, by Propositions \ref{prop:poisson_independence_under_taking_subtree} and \ref{prop:poisson_independence_under_taking_subsegment}, we have that
    \begin{equation*}
        \E[\cost(\plan(\sigma), T)] = \sum_{i = 1}^k \frac{\tau}{\hat{p}_i} \E[\cost(\plan(\sigma_i), T_i)] = \sum_{i = 1}^k \frac{\tau}{\hat{p}_i} \cdot 2 w(T_i) = 2 \sum_{i = 1}^k \frac{\tau}{\hat{p}_i} w(T_i).
    \end{equation*}

    Now, let us lower bound the expected cost for the optimal offline schedule for $\sigma$. 
    By Proposition \ref{prop:poisson_independence_under_taking_subtree}, we have that
    \begin{equation*}
        \E[\cost(\opt(\sigma), T)] = \sum_{i = 1}^k \E[\cost(\opt(\sigma|_{T_i}), T_i)] \geq \sum_{i = 1}^k \frac{\tau}{p_i} \frac{3}{16} w(T_i).
    \end{equation*}
    Since by definition of Algorithm \ref{alg:plan_2}, it holds that $p_i < 2\hat{p}_i$ for $i \in [k]$, we can rewrite the above as
    \begin{equation*}
        \E[\cost(\opt(\sigma), T)] > \sum_{i = 1}^k \frac{\tau}{2\hat{p}_i} \frac{3}{16} w(T_i) = \frac{3}{32} \sum_{i = 1}^k \frac{\tau}{\hat{p}_i} w(T_i).
    \end{equation*}

    Thus, the ratio between the expected costs of $\opt$ and $\plan$ algorithms is upper bounded by $64/3$. 
    However, to simplify the calculations, until now we assumed that the time horizon $\tau$ is a multiple of $2\hat{p}_k$. 
    Nonetheless, this implies that $\plan$ achieves a ratio of expectations equal to $64/3$, since with the value of $\tau$ going to infinity, the marginal contribution of extra cost generated by the last service (line \ref{line:plan_last_extra_service} of Algorithm \ref{alg:plan_3}) goes to 0. 
\end{proof}

\section{General instances}
\label{section:general}
In this section, we describe our algorithm $\general$ for an arbitrary stochastic instance ($T,\gamma, \llambda, \textbf{w}$) and prove that it achieves a constant RoE. 
The main idea is to distinguish two types of requests, and apply a different strategy for each type. 
The first type are the requests that are located close to the root. 
These requests will be served immediately at their arrival times, i.e., we apply $\instant$ to the corresponding sub-sequence. 
The second type includes all remaining requests and they are served in a periodic manner. To determine the period of these vertices, we will use the algorithm $\plan$ on a specific heavy instance $(T', \llambda^h)$ 
constructed in Section \ref{section:construction_light_and_heavy}. The construction of this heavy instance relies on a partition of the vertices of $T$ into \emph{balanced parts}, whose definition and construction is given in Section \ref{section:balanced_partition}.  Intuitively, a part is balanced when it is light (or close to being light), but if we merge all vertices of the part into a single vertex whose weight corresponds to the average distance to the root of the part, then we obtain an heavy edge. This ``merging'' process is captured by the construction of the \emph{augmented tree} $T'$, which is part of the heavy instance. The augmented tree is essentially a copy of $T$ with the addition of one (or two) new vertices for each balanced part. See Section \ref{section:construction_light_and_heavy} for the formal description. 

Once we have determined the corresponding heavy instance, we can compute the periods of each vertex of the heavy instance with the $\plan$ algorithm. The period of a vertex in the original instance is equal to the period of the corresponding vertex in the heavy instance. The full description of the $\general$ algorithm is given in Section \ref{section:general_alg}. 

The main challenge of this section is to analyze the ratio of expectations of this algorithm, and in particular to establish good lower bounds on the expected cost of the optimal offline schedule. This is done in Section \ref{section:upper_bound_gen}, where we prove two lower bounds (Lemmas \ref{lemma:low_bound_OPT_1} and \ref{lemma:low_bound_OPT_2}) that depend on the heavyness of each part of the balanced partition. 

In the entire section, we assume without loss of generality that $\gamma$ has only one child. To see that this is possible, consider an MLA instance $(T, \sigma)$ with tree $T$ rooted at vertex $\gamma$ that has at least $k \geq 2$ children. 
For each $i \in [k]$ let $T_i$ be a tree obtained from $T$ by pruning all the children of $\gamma$ except the $i$th one, and denote the sequence of requests in $\sigma$ arriving at $T_i$ by $\sigma_i$. 
Finding a schedule for instance $(T, \sigma)$ is then equivalent to finding a family of schedules for $(T_i, \sigma_i)$ since each service $s$ for $(T, \sigma)$ can be partitioned into a set of services for each instance $(T_i, \sigma_i)$. Further the sum of the weights of the new services is the same as the weight of $s$. 

Moreover, recall that we assume $\llambda(\gamma) = 0$, as each request arriving at $\gamma$ can be served immediately with zero weight cost. 

\paragraph{Notations.}
Recall that given the edge-weighted tree $T$ rooted at $\gamma$ and a set of vertices $U \subseteq V(T)$, $T[U]$ denotes the forest induced on $U$ in $T$.
We say that a subset $U \subseteq V(T)$ is \emph{connected} if $T[U]$ is connected (i.e., $T[U]$ is a subtree of $T$ but not a forest).
If $U \subseteq V(T)$ is connected, we write $\gamma(U) = \gamma(T[U])$ to denote the root vertex of $T[U]$, i.e., the vertex in $U$ which has the shortest path length to $\gamma$ in the original tree $T$. 
Given any vertex $u \in V(T)$, let $V_u \subseteq V(T)$ denote all the descendant vertices of $u$ in $T$ (including $u$). 
For simplification, set $w_{\gamma(T)} = \infty$.  
Given $T = (V,E)$ and $\llambda: V(T) \rightarrow \R_+$, and $U \subseteq V$, we denote $\llambda \vert_U: U \rightarrow \R_+$ such that $\llambda \vert_U(u) = \llambda(u)$ for each $u \in U$. 
For a sequence of requests $\sigma \sim (T, \llambda)$, we use $\sigma \vert_U = \{r \in \sigma: \ell(r) \in U\}$ to denote the corresponding sequence for $T[U]$. 

\subsection{Balanced partition of $V(T)$}
\label{section:balanced_partition}
Given $\llambda: V(T) \rightarrow \R_+$, recall that
    $\pi(T,\llambda) = \sum_{u \in V(T)} \llambda(u) \cdot d(u,\gamma(T))$. 
When the context is clear we simply write $\pi(T) = \pi(T, \llambda)$, and for a connected subset $U\subseteq V(T)$ we simply write $\pi(U):=\pi(T[U],\llambda|_U)$. 
\begin{definition} \label{definition:balanced_part}
    Given a stochastic instance $(T,\llambda)$, we say that $U \subseteq V(T)$ is \emph{balanced} if $U$ is connected and if one of the following conditions holds:
    \begin{description}
        \item[(1)] $\pi(U) \le 1$, and either $\gamma(U)=\gamma(T)$ or $\pi(U \cup \{\parent{\gamma(U)}\}) > 1$. 
        In this case, we say that $U$ is of \emph{type-I}.
        \item[(2)] $\pi(U) > 1$, and for each child vertex $y$ of $\gamma(U)$ in $T[U]$, we have $\pi(\{\gamma(U)\}\cup(U \cap V_y)) < 1$. In this case, we say that $U$ is of \emph{type-II}.
    \end{description}
\end{definition}
\noindent \emph{Remark.} Note that the root $\gamma(U)$ of a balanced part $U$ of type-II is necessarily to have at least two children in $T[U]$. 

\begin{figure}
    \centering
    \includegraphics[width=0.8\textwidth,page=1]{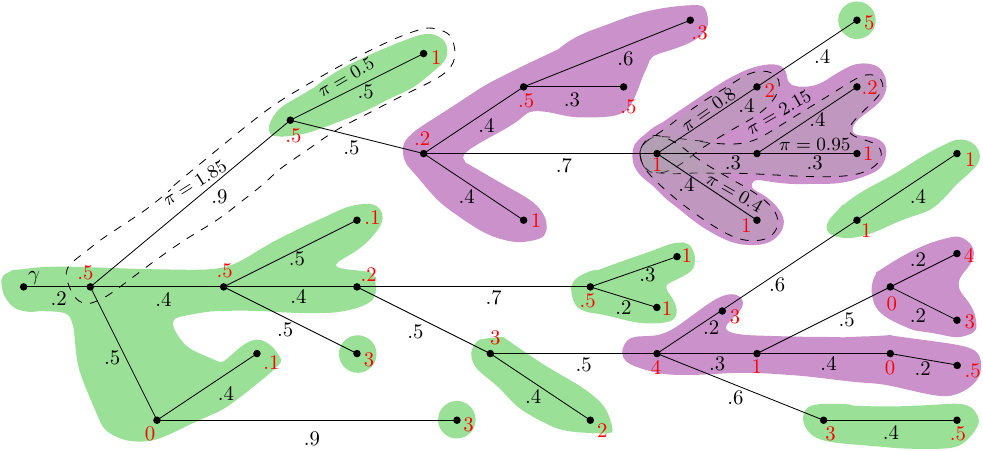}
    \caption{An example of a balanced partition (Definition \ref{definition:balanced_partition}). The weight of each edge is shown in black, and the arrival rate of each vertex is shown in red. Green subsets corresponds to parts of type-I while purple ones correspond to parts of type-II. Some value of $\pi$ are shown for the top-left type-I part and for the top-right type-II part.  }
    \label{fig:example_partition}
\end{figure}

\begin{definition} \label{definition:balanced_partition}
    Given a stochastic instance $(T,\llambda)$ and a partition $\mathcal{P}$ of the vertices $V(T)$, we say that $\mathcal{P}$ is a \emph{balanced partition} of $T$ if every part $U\in\mathcal{P}$ is balanced. 
\end{definition}

See Figure \ref{fig:example_partition} for an example of a balanced partition. If $\mathcal{P}$ is a \emph{balanced partition} of $T$, then the part $U \in \mathcal{P}$ containing $\gamma(T)$ is called the \emph{root part} in $\mathcal{P}$. 
Since we assumed that $\gamma(T)$ has only one child vertex, we deduce from the previous remark that the root part is necessarily of type-I. 
Given a balanced partition $\mathcal{P}$, we denote $\mathcal{P}^*:=\mathcal{P}\setminus \{\gamma(\mathcal{P})\}$, we denote $\mathcal{P}_1\subseteq\mathcal{P}$ the set of parts of type-I, $\mathcal{P}_1^*:=\mathcal{P}_1\cap \mathcal{P}^*$, and $\mathcal{P}_2\subseteq\mathcal{P}$ the set of parts of type-II. 

\begin{lemma} \label{lemma:balanced_partition}
    Given any stochastic instance $(T, \llambda)$, there exists a balanced partition of $T$.
    Moreover, such a partition can be computed in $O(|V(T)|^2)$ time.\footnote{we assume that basic operations on numbers can be done in constant time. }
\end{lemma}

\begin{proof}
We first describe an algorithm that constructs a partition $\mathcal{P}$, and then argue that this partition is balanced. 
See Figure \ref{fig:proof_lemma_partition} for a visual support. 
We sort the vertices $u_1,\dots, u_n$ of $T$ (where $n=|V(T)|$) by decreasing distances to the tree root, i.e., for $1 \le i < j \le n$, $d(u_i,\gamma(T)) \ge d(u_j,\gamma(T))$. 
Let $\mathcal{P}^{(0)} = \emptyset$. 
For each $i \in [n]$, do the following. 
Let $C_i\subseteq \{1,\dots, n-1\}$ be the subset of indexes $j$, such that (i) $u_j$ is a child of $u_i$; (ii) $u_j \notin \bigcup_{U \in \mathcal{P}^{(i-1)}}U$. 
Define $U_i := \{u_i\} \cup \bigcup_{j \in C_i} U_j$ recursively. 
If $i = n$ (i.e., if $u_i$ is the root of $T$) or $\pi(U_i\cup \{\parent{u_i}\}) > 1$, then define $\mathcal{P}^{(i)} := \mathcal{P}^{(i-1)} \cup \{U_i\}$. 
Otherwise, define $\mathcal{P}^{(i)} := \mathcal{P}^{(i-1)}$. 

We first show that for each $i \in [n]$, the set $U_i$ is connected. 
This shows the correctness of the algorithm\footnote{$\pi(\cdot)$ is only defined for connected sets. } and will also be useful later. 
We proceed by induction. 
Given $i \in [n]$, suppose that all $U_j$ with $j < i$ are connected. 
Since the vertices are ordered by decreasing distance to the root of $T$, it holds that for every $j \in C_i$, we have $j < i$, and hence $U_j$ is connected. 
The set $U_i = \{u_i\} \cup \bigcup_{j \in C_i} U_j$ is connected since for each $j \in C_i$, $u_j$ and $u_i$ are adjacent in $T$, $u_j \in U_j$, and $U_j$ is connected.  

Let $\mathcal{P} = \mathcal{P}^{(n)}$ be the set of subsets of $V(T)$ at the end of the algorithm. 
Now we show that $\mathcal{P}$ is a partition of $V(T)$.
For each $i \in [n]$, let $b(i)$ be the smallest integer $j$ such that (i) $u_j$ is on the path from $\gamma(T)$ to $u_i$; (ii) $\mathcal{P}^{(j)} \neq \mathcal{P}^{(j-1)}$. 
Note that such index $b(i)$ exists and is unique.
Besides, for each $i \in [n]$, $u_i \in U_{b(i)}$ and $U_{b(i)} \in \mathcal{P}$. 
Hence, $\mathcal{P}$ is indeed a partition of $V(T)$. 

\begin{figure}
    \centering
    \includegraphics[width=0.8\textwidth, page=2]{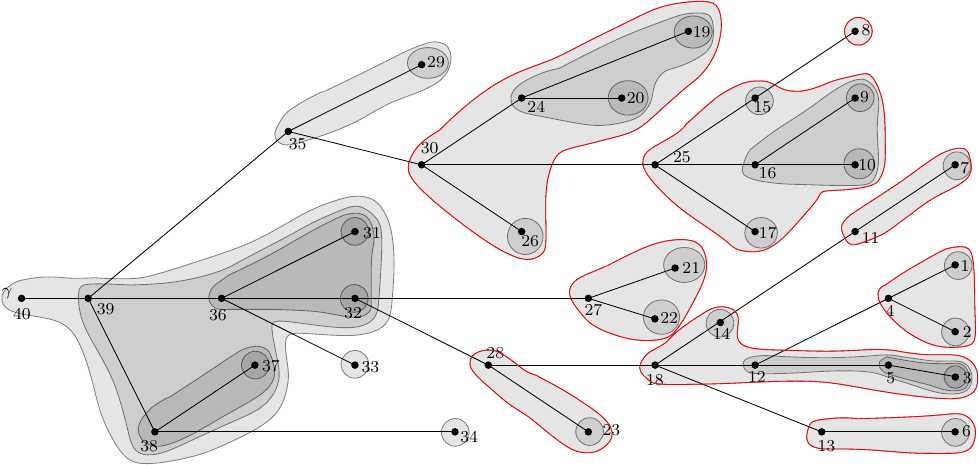}
    \caption{Construction of a balanced partition in the proof of Lemma \ref{lemma:balanced_partition}. The weights of the edges and the arrival rates of the vertices are the same as in Figure \ref{fig:example_partition}. The numbers represent an ordering of the vertices. The gray sets corresponds to $U_i$, for $i\in [40]$. 
    We illustrate the step $i=30$ of the algorithm. We have $C_{30}=\{24,26\}$ and $U_{30}=\{u_{30}\}\cup U_{24}\cup U_{30}=\{u_{30},u_{24},u_{19}, u_{20},u_{26}\}$. Since $\pi(U_{30}\cup \{u_{25}\})=2.65>1$, we add $U_{30}$ into $\mathcal{P}^{(29)}$ to create $\mathcal{P}^{(30)}$ (red stoke). We remark that $U_{30}$ is a balanced subset of type-II and $U_{27}$ is a balanced subset of type-II, while $U_{25}$ is not a balanced subset. }
    \label{fig:proof_lemma_partition}
\end{figure}

It remains to show that each part $W \in \mathcal{P}$ is either of type-I or of type-II. Let $i$ be the index such that $U_i = W$. 
Note that $\mathcal{P}^{(i)} \neq \mathcal{P}^{(i-1)}$ by the definition of our algorithm. 
This implies that either $u_i = \gamma(T)$ or $\pi(U_i \cup \{\parent{u_i}\}) > 1$. 
If $\pi(W) \le 1$, then $W$ is of type-I. 
Otherwise, $\pi(W) > 1$. 
We show that in this case $W$ is of type-II, i.e., for each child vertex $y$ of $\gamma(W)$ in $T[W]$, we have $\pi(\{\gamma(W)\}\cup (W\cap V_y)) < 1$. 
Notice that $\gamma(W) = u_i = \parent{y}$, and that for each child vertex $y$ of $u_i$ in $T[W]$, there is an index $j \in C_i$ such that $y = u_j$. 
Furthermore, by construction we have $W \cap V_{u_j} = U_j$. 
Thus, $\{\gamma(W)\} \cup (W \cap V_y) = U_j \cup \{\parent{u_j}\}$. 
Now, by definition of $C_i$, we know that $u_j \notin \bigcup_{U\in \mathcal{P}^{(i-1)}}U$, and in particular, $u_j \notin \bigcup_{U\in \mathcal{P}^{(j)}}U$, since $j \le i-1$. 
This implies that $\mathcal{P}^{(j)} = \mathcal{P}^{(j-1)}$, and hence that $\pi(U_j \cup \{\parent{u_j}\}) \le 1$, which is what we wanted to prove. 
To summerize, $\mathcal{P}$ is indeed a balanced partition. 

For the running time to produce $\mathcal{P}$, recall that $n = |V(T)|$ denotes the number of vertices. 
On one hand, sorting all the tree vertices in the order of their distances to the tree root takes a running time at most $O(n^2)$. 
On the other hand, there are $n$ iterations for producing the balanced partition $\mathcal{P}$. 
In each iteration $i \in [n]$, determining $C_i$ needs $O(n)$ time and calculating $\pi(U_i \cup \{\parent{u_i}\})$ needs $O(n)$ time. 
As such, the total running time is $O(n^2)$. 

This concludes the proof of Lemma \ref{lemma:balanced_partition}.
\end{proof}

\subsection{The heavy instance}
\label{section:construction_light_and_heavy}

Given a stochastic instance $(T,\llambda)$, and a balanced partition $\mathcal{P}$ of $T$, we construct a tree $T'$ that we call the \emph{augmented tree} of $T$. 
This tree is essentially a copy of $T$ with additional one or two vertices for each part of $\mathcal{P}^*$.\footnote{Recall that $\mathcal{P}^* = \mathcal{P} \setminus \{\gamma(\mathcal{P})\}$. 
Here $\gamma(\mathcal{P})$ denotes the particular part in $\mathcal{P}$ which includes the tree root $\gamma(T)$. } 
Then, we define arrival rates $\llambda^h$ on $T'$ in a way that the stochastic instances $(T',\llambda^h)$ is heavy. 
Finally, we construct from a request sequence $\sigma$, the corresponding \emph{heavy sequence} $\sigma^h$ for the augmented tree. 

\paragraph{Construction of the augmented tree.} 
We define $T' = (V', E')$ where
\begin{equation*}
    V' = V(T) \cup \{z_U, z'_U: U \in \mathcal{P}^*\}, 
\end{equation*}
and the edge set $E'$ is constructed based on $E(T)$ as follows (see Figure \ref{fig:construction_T'}). 
\begin{itemize}
    \item[-] First, for each $U \in \mathcal{P}_1^*$, replace the edge $(\gamma(U), \parent{\gamma(U)})$ (of length $w_{\gamma(U)}$) by two edges $(\gamma(U), z'_U)$ and $(z'_U,\parent{\gamma(U)})$ of respective lengths $(1 - \pi(U)) / \llambda(U)$ and $w_{\gamma(U)} - (1 - \pi(U)) / \llambda(U)$; where $\parent{\gamma(U)}$ denotes the parent of $\gamma(U)$ in $T$. Then, add an edge $(z_U, z'_U)$ of weight $1 / \llambda(U)$. 
    \item[-] For each $U \in \mathcal{P}_2$, set $z'_U=\gamma(U)$, and add an edge $(z_U, z'_U)$ of weight $\pi(U) / \llambda(U)$. 
\end{itemize}
This completes the construction of the augmented tree. 
Notice that if a part $U=\{u\}$ in $\mathcal{P}$ that contains only one vertex, then we have $\pi(U)=0$ and thus $U$ is necessarily of type-I. 
To simplify, in the following we identify vertices in $T$ with their copy in $T'$, and consider that $V(T)$ is a subset of $V(T')$. 

\begin{figure}
    \centering
    \includegraphics[width=0.8\textwidth,page=3]{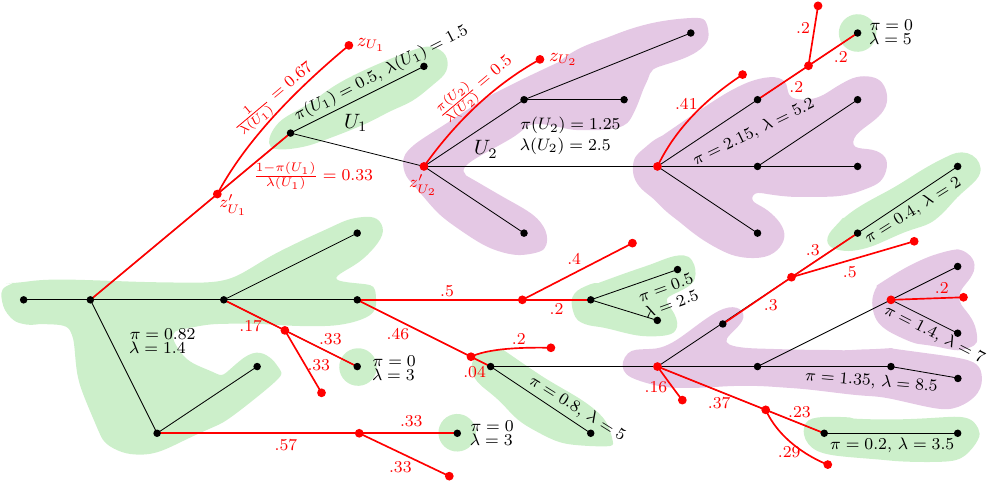}
    \caption{The construction of the augmented tree associated with the instance and the balanced partition of Figure \ref{fig:example_partition}. The new edges and vertices are shown in red. The illustrate the calculation of the length of these edges for a part $U_1$ of type-I and for a part $U_2$ of type-II. For each part $U$ of the partition, we indicate the values of $\llambda(U)$ and $\pi(U)$. For simplicity, we have rounded the values to their second decimal. }
    \label{fig:construction_T'}
\end{figure}

\paragraph{Arrival rates for the heavy instance.}
Recall that $\mathcal{P}^* = \mathcal{P} \setminus \{\gamma(\mathcal{P})\}$, where $\gamma(\mathcal{P})$ denotes the part in $\mathcal{P}$ containing the root $\gamma(T)$. 
We define $\llambda^h: V(T') \rightarrow \R_+$ as follows: for each $U \in \mathcal{P}^*$, set $\llambda^h(z_U) = \llambda(U)$; and $\llambda^h(u) = 0$ otherwise. 

\begin{proposition} \label{proposition:heavy_instance_is_heavy}
    $(T',\llambda^h)$ is heavy.
\end{proposition}

\begin{proof}
For each $U\in\mathcal{P}^*_1$ we have $\llambda^h(z_U)\cdot w(z_U)=\llambda(U)\cdot (1/\llambda(U))=1$. 
For each $U\in\mathcal{P}^*_2$ we have $\llambda^h(z_U)\cdot w(z_U)=\llambda(U)\cdot (\pi(U)/\llambda(U))=\pi(U)>1$. Hence $(T',\llambda^h)$ is heavy; see Definition \ref{definition:heavy_instance}. 
\end{proof}

\paragraph{The heavy sequence.}

\begin{definition}
    Given a stochastic instance $(T,\llambda)$, a balanced partition $\mathcal{P}$ of $T$, the corresponding augmented tree $T'$, and a sequence of request $\sigma \sim (T,\llambda)^\tau$, we construct the \emph{heavy sequence associated with} $\sigma$ for $T'$ and denoted by $\sigma^h$ as follows: for each request $r = (t, u) \in \sigma$ that is located on some part $U\in \mathcal{P}^*$ (i.e., $u \in U$), there is a request $(z_U,t)$ in $\sigma^h$.
    \label{definition:heavy_sequence}
\end{definition}

\noindent
\emph{Remark.} It is important to notice that $\sigma^h$ can be constructed in an online fashion: for any time $t$, the restriction of the $\sigma^h$ to the requests that arrives before $t$ only depends on the requests that arrives before $t$ in $\sigma$.

\begin{proposition} \label{proposition:correspondance_lambdas}
    Given a stochastic instance $(T, \llambda)$ and the corresponding heavy instance $(T', \llambda^h)$, let $\sigma' \sim (T', \llambda^h)^\tau$ denote a random sequence for $T'$ and let $X(\sigma')$ be a random variable depending on $\sigma'$. 
    Then, for any $\tau > 0$, it holds that 
    \begin{equation*}
        \E\left[X(\sigma') \mid \sigma' \sim (T',\llambda^h)^\tau\right] = \E\left[X(\sigma^h) \mid \sigma \sim (T,\llambda)^\tau\right],
    \end{equation*}
    where $\sigma^h$ is the heavy sequence associated with $\sigma \sim (T, \llambda)^\tau$ (see Definition \ref{definition:heavy_sequence}).
\end{proposition}

\begin{proof}
Suppose we are given a random sequence of requests $\sigma \sim (T,\llambda)^\tau$. 
We show that $\sigma^h\sim(T',\llambda^h)^\tau$. 
According to Proposition \ref{prop:poisson_model_centralized}, (i) the waiting time between two consecutive requests in $\sigma$ follows exponential distribution $\expdistr{\llambda(T)}$; (ii) once a request $r \in \sigma$ arrives, the probability that it is located at a vertex $u \in V(T)$ is $\llambda(u) / \llambda(T)$.
This implies that (i) the waiting time between any two consecutive requests arriving at vertices $\bigcup_{U \in \mathcal{P}^*} U$ follows exponential distribution $\expdistr{\llambda(\mathcal{P}^*)}$;\footnote{For simplicity, denote $\llambda(\mathcal{P}^*) := \sum_{U \in \mathcal{P^*}} \sum_{u \in U} \llambda(u)$ in this section.}
(ii) the probability that a request is located at some vertex of $U\in \mathcal{P}^*$ is $\llambda(U) / \llambda(\mathcal{P}^*)$. 
Note that when a request $r \in \sigma$ arrives at $\bigcup_{U \in \mathcal{P}^*} U$, then a corresponding request in $\sigma^h$ arrives at vertex $z_U \in V(T')$. 
As a result, we know that $\sigma^h$ follows the centralized Poisson arrival model (see Proposition \ref{prop:poisson_model_centralized}).
Due to the equivalence between the centralized and distributed Poisson arrival model, we thus have $\sigma^h \sim (T', \llambda^h)^{\tau}$.
Since $\sigma' \sim (T', \llambda^h)^{\tau}$ also follows Poisson arrivals model generated over $T'$ with arrival rates $\llambda^h$, we thus have $\E\left[X(\sigma') \mid \sigma \sim (T',\llambda^h)^\tau\right] = \E\left[X(\sigma^h) \mid \sigma \sim (T,\llambda)^\tau\right]$. 
\end{proof}

\subsection{The algorithm}
\label{section:general_alg}
The algorithm GEN works as follows. 
It is given a stochastic instance $(T, \llambda)$, known in advance, and a sequence of requests $\sigma$ for $T$, revealed over time. 
In the pre-processing step, the algorithm computes a balanced partition $\mathcal{P}$ of $T$ (Lemma \ref{lemma:balanced_partition}), and computes a light instance $(T[\gamma(\mathcal{P})], \sigma|_{\gamma(\mathcal{P})})$ 
and the heavy instance $(T', \sigma^h)$ (Section \ref{section:construction_light_and_heavy}). 
Upon arrival of each request, the algorithm updates the sequences of requests $\sigma|_{\gamma(\mathcal{P})}$ and $\sigma^h$ as described in the previous paragraph. 

The algorithm runs PLAN (Algorithm \ref{alg:plan}) on input $(T',\sigma^h)$. 
Suppose that PLAN serves at time $t$ a set of vertices $\{z_U, U\in \mathcal{P}'\}\subseteq V(T')$ for some subset $\mathcal{P}'\subseteq \mathcal{P}^*$. Then, GEN serves at time $t$ all pending requests on vertices $\left(\bigcup_{U\in \mathcal{P}'}U\right)\subseteq V(T)$. 

In parallel, the algorithm runs INSTANT (Algorithm \ref{pseudocode:instant}) on input $(T[\gamma(\mathcal{P})], \sigma|_{\gamma(\mathcal{P})})$, and performs the same services. 
This finishes the description of the algorithm GEN.

\RestyleAlgo{boxruled}
\LinesNumbered
\SetAlgoVlined
\begin{algorithm}
\caption{GEN}
\label{pseudocode:gen}
\KwIn{stochastic instance $(T, \llambda)$ and a random sequence of requests $\sigma \sim (T, \llambda)^\tau$}
\KwOut{a valid schedule of $\sigma$}
----- pre-processing the given instance ----- \\
produce a balanced partition $\mathcal{P}$ for $T$ (see Lemma \ref{lemma:balanced_partition}); \\
construct the heavy instance $(T', \llambda^h)$ (according to Section \ref{section:construction_light_and_heavy}); \\
use PLAN (Algorithm \ref{alg:plan}) to determine the period of the vertices of $T'$; \\
----- Serve the requests ----- \\
\For{each request $r \in \sigma$}{
    \If{$r$ arrives in $\gamma(\mathcal{P})$}{
        serve $r$ immediately. 
    }
    \If{$r$ arrives in a vertex of $U \in \mathcal{P}^*$}{
        serve $r$ at time $t(r')$ where $r'\in \sigma^h$ is the corresponding request located on $z_U$ and $t(r')$ is the time at which $r'$ is served by $\plan(\sigma^h)$.
    }
}
\end{algorithm}

\begin{observation}
$\gen(\sigma)$ is a valid schedule for any sequence of requests $\sigma$. 
\end{observation}
\begin{proof}
Each time a request $r \in \sigma$ arrives, if $r$ is classified into light sequence $\sigma|_{\gamma(\mathcal{P})}$, then it is served immediately; if $r$ is classified into heavy sequence $\sigma^h$, then it is served periodically.
\end{proof}

\subsection{Analysis on GEN's ratio of expectations}
\label{section:general_analysis}
In this section, we analyze the ratio of expectations of GEN for arbitrary stochastic instances and show that it can be bounded by 210. 
We first present in Section \ref{section:lower_bounds_opt} two lower bounds on the expectation of optimal schedule (Lemma \ref{lemma:low_bound_OPT_1} and Lemma \ref{lemma:low_bound_OPT_2}), then in Section \ref{section:upper_bound_gen}, we upper bound the expected cost of the schedule produced by our algorithm GEN (Lemma \ref{lemma:bound_GEN}), and finally we combine the three results in Section \ref{section:proof_main_theorem} to prove GEN achieves a constant RoE (Theorem \ref{theorem:general_any}). 

In this section, we use the following notation. 
Given a balance partition $\mathcal{P}$ of a stochastic instance $(T, \llambda)$, for each part $U \in \mathcal{\mathcal{P}}$ we set $\pi'(U) = 1$, if $U\in \mathcal{P}^*_1$, and $\pi'(U)=\pi(U)$, otherwise. 

\subsubsection{Lower bounds on the cost of the optimal schedule}
\label{section:lower_bounds_opt}

\begin{lemma} \label{lemma:low_bound_OPT_1}
    Given a stochastic instance $(T,\llambda)$,  a balanced partition $\mathcal{P}$ for this instance, and $\tau > 0$, it holds that 
    \begin{equation*}
        \E\big[\cost(\opt(\sigma), T)\big] \ge \frac{3}{16}(1-e^{-1}) \cdot \tau \cdot \left(\sum_{U \in \mathcal{P}} \pi'(U)\right),
    \end{equation*}
    where $\opt(\sigma)$ denotes an optimal schedule for $\sigma$ and the expectation is taken over all random sequences $\sigma\sim(T,\llambda)^\tau$. 
\end{lemma}

\begin{proof}
To prove this bound, we define a family $\widetilde{\mathcal{T}}$ of edge-disjoint light subtrees of $T'$.  This family $\widetilde{\mathcal{T}}$ is defined as follows:
\begin{itemize}
    \item[0.] we add $T'[\gamma(\mathcal{P})]$;
    \item[1.] for each  $U\in \mathcal{P}^*_1$, we add the subtree $T'[U\cup \{z'_U\}]$;
    \item[2.] for each $U\in \mathcal{P}_2$, and each child $y$ of $\gamma(U)$ in $T'[U]$, we add $T'[(U\cap V_y)\cup \{\gamma(U)\}]$;    
\end{itemize}
Notice that these subtrees are pairwise edge-disjoint, but two subtrees that corresponds to the same part of type-II share the same root vertex. 

For each subtree $\widetilde{T}\in \widetilde{\mathcal{T}}$ we construct an arrival rate function $\llambda_{\widetilde{T}}:V(\widetilde{T})\rightarrow \R_+$. 
We fix an arbitrary strict total order on the vertices of $T$. 
In cases 0 and 1, we simply define $\llambda_{\widetilde{T}} :=  \llambda|_{\widetilde{T}}$. 
In case 2, if $\mathcal{T}$ corresponds to a part $U\in\mathcal{P}_2$ and a vertex $y$ that is \emph{not} the smallest\footnote{according to the fixed order.} child of $\gamma(U)$, then we set $\llambda_{\widetilde{T}}(\gamma(U)) = 0$ and $\llambda_{\widetilde{T}}(u) = \lambda(u)$ for $u \in V(\widetilde{T}) \setminus \{\gamma(U)\}$. 
If $y$ is the smallest child of $\gamma(U)$, then we simply set $\llambda_{\widetilde{T}} :=  \llambda|_{\widetilde{T}}$. 

We claim that for each $\widetilde{T} \in \widetilde{\mathcal{T}}$, the stochastic instance $(\widetilde{T}, \llambda_{\widetilde{T}})$ is light (see Definition \ref{definition:light_instance}). 
For subtrees that are associated with parts of type-II, or with the root part, this simply follows from the definition of balanced parts (Definition \ref{definition:balanced_part}). 
Now consider a part $U \in \mathcal{P}^*_1$ (of type-I), and its associated subtree $\widetilde{T} = T'[U \cup \{z'_U\}]$. 
Notice that $z'_U$ is the root in $\widetilde{T}$ and its only child in $\widetilde{T}$ is $\gamma(U)$. 
Thus, for each $u \in U$, we have $d(u, z'_U) = d(u, \gamma(U)) + w(\gamma(U), z'_U) = d(u, \gamma(U)) + (1 - \pi(U)) / \llambda(U)$. 
We now calculate $\pi(\widetilde{T}, \llambda_{\widetilde{T}})$ and show that is is equal to 1, which means that $(\widetilde{T}, \llambda_{\widetilde{T}})$ is light:
\begin{align*}
    \pi(\widetilde{T}, \llambda_{\widetilde{T}}) &= \llambda(z'_U) \cdot d(z'_U, z'_U) + \sum_{u \in U} \llambda(u) \cdot d(u,z'_U)
    =0+ \sum_{u \in U} \llambda(u)\left(d(u,\gamma(U)) + \frac{1 - \pi(U)}{\llambda(U)}\right)\\
    &= \pi(U) + \llambda(U) \cdot \frac{1 - \pi(U)}{\llambda(U)} = 1.
\end{align*}
Thus, by Theorem \ref{lemma:opt_light_lower-bound}, for each $\widetilde{T}\in \widetilde{\mathcal{T}}$,
\begin{equation}
    \label{eq:proof_1}
    \E[\cost(\opt(\sigma'),\widetilde{T})]\ge \frac{3}{16}(1-e^{-1}) \cdot \tau \cdot\pi(\widetilde{T},\llambda|_{\widetilde{T}}),
\end{equation}
where $\opt(\sigma')$ denote the optimal schedule for the sequence of request $\sigma'$ in $T'$, and the expectation is taken over all the random sequences $\sigma'\sim (\widetilde{T},\llambda|_{\widetilde{T}})^\tau$. 

Now, let $\sigma$ be a sequence of requests for $T$ of duration $\tau$. 
For each $\widetilde{T} \in \widetilde{\mathcal{T}}$, we define a request sequence $\sigma_{\widetilde{T}}$ for $\widetilde{T}$: for each request $(t, u) \in \sigma$, there is a request $(t, u)$ in $\sigma_{\widetilde{T}}$ if and only if $u \in V(\widetilde{T})$ and $\llambda_{\widetilde{T}}(u) > 0$. 

Let $S = \opt(\sigma)$ denote an optimal schedule for $\sigma$. 
For each $\widetilde{T} \in \widetilde{\mathcal{T}}$, we define a schedule $S_{\widetilde{T}}$ for $\sigma_{\widetilde{T}}$ as follows. 
For each service $s = (R(s), t) \in S$, there is a service $s_{\widetilde{T}} := (R(s) \cap V(\widetilde{T}), t)$ in $S_{\widetilde{T}}$. 

It is not difficult to see that
\begin{itemize}
    \item[-] for each $\widetilde{T} \in \widetilde{\mathcal{T}}$, $S_{\widetilde{T}}$ is a valid schedule for $\sigma_{\widetilde{T}}$, and particularly $\cost(S_{\widetilde{T}}, \sigma_{\widetilde{T}}) \ge \cost(\opt(\sigma_{\widetilde{T}}),\widetilde{T})$. 
    \item[-] $\delay(S) = \sum_{\widetilde{T} \in \widetilde{\mathcal{T}}} \delay(S_{\widetilde{T}})$. 
\end{itemize}

We now argue that $\weight(S,T) \ge \sum_{\widetilde{T} \in \widetilde{\mathcal{T}}} \weight(S_{\widetilde{T}}, {\widetilde{T}})$. 
Indeed, since subtrees in $\widetilde{\mathcal{T}}$ are pairwise edge-disjoint, it holds that for each service $s \in S$, we have:
$\weight(s,T) \ge \sum_{\widetilde{T} \in \widetilde{\mathcal{T}}} \weight(s_{\widetilde{T}},\widetilde{T})$, which implies what we want.  

Finally, we show that $\sum_{\widetilde{T} \in \widetilde{\mathcal{T}}} \pi(\widetilde{T},\llambda_{\widetilde{T}}) = \sum_{U \in \mathcal{P}} \pi'(U)$. 
Let us first consider the root part, and its associated subtree $\widetilde{T} = T'[\gamma(\mathcal{P})]$. 
We have 
\begin{equation*}
    \pi(\widetilde{T},\llambda_{\widetilde{T}}) = \pi(T[\gamma(\mathcal{P})],\llambda) = \pi'(\gamma(\mathcal{P})).
\end{equation*}
Now, consider a part $U\in \mathcal{P}_2$ (of type-II). 
Let $Y \subseteq V(T)$ denote the children of $\gamma(U)$. 
The part $U$ is associated with subtrees $T[(U\cap V_y) \cup \{\gamma(U)\}]$, for $y \in Y$. 
The family $\{U \cap V_y, y \in Y\}$ forms a partition of $U$, and thus we have 
\begin{align*}
    \pi'(U) = \pi(U) = \sum_{u \in U}\llambda(u) \cdot d(u,\gamma(U)) &= \sum_{y \in Y} \sum_{u \in U \cap V_y} \llambda_{\widetilde{\mathcal{T}}}(u) \cdot d(u,\gamma(U))\\
    &=\sum_{u\in U\cap V_y}\pi(T[(U\cap V_y)\cup \{\gamma(U)\}],\llambda_{\widetilde{\mathcal{T}}}).
\end{align*}

Given $U \in \mathcal{P}^*_1$, let $\widetilde{T}$ be the subtree associated. 
We have proved before that in this case $\pi(\widetilde{T}, \llambda_{\widetilde{T}}) = 1 = \pi'(U)$ by definition of $\pi'$. 
Finally, we have 
\begin{align*}
    \cost(\opt(\sigma),T)&=\cost(S,T)=\delay(S)+\weight(S,T)\ge \sum_{\widetilde{T}\in \widetilde{\mathcal{T}}} \delay(S_{\widetilde{T}}) + \sum_{\widetilde{T}\in \widetilde{\mathcal{T}}} \weight(S_{\widetilde{T}},{\widetilde{T}})\\
    &=\sum_{\widetilde{T}\in \widetilde{\mathcal{T}}} \cost(S_{\widetilde{T}},{\widetilde{T}})\ge \sum_{\widetilde{T}\in \widetilde{\mathcal{T}}} \cost(\opt(\sigma_{\widetilde{T}}),{\widetilde{T}}).
\end{align*}
We now take expectation over all the random sequences $\sigma \sim (T, \llambda)^\tau$. 
It is not difficult to see (the proof is similar as the proof of Proposition \ref{proposition:correspondance_lambdas}) that for $\widetilde{T} \in \widetilde{\mathcal{T}}$, $\sigma \sim (T,\llambda)^\tau \Longrightarrow \sigma_{\widetilde{T}} \sim (\widetilde{T},\llambda_{\widetilde{T}})^\tau$. Thus, 
\begin{align*}
    \E\Big[\cost(\opt(\sigma),T) \mid \sigma \sim (T,\llambda)^\tau\Big] &= \sum_{\widetilde{T} \in \widetilde{\mathcal{T}}} \E\Big[\cost(\opt(\sigma_{\widetilde{T}}),\widetilde{T}) \mid \sigma_{\widetilde{T}}\sim(\widetilde{T},\llambda_{\widetilde{T}})^\tau\Big]\\
    &\ge \sum_{\widetilde{T}\in \widetilde{\mathcal{T}}}\frac{3}{16}(1-e^{-1})\tau \cdot\pi(\widetilde{T},\llambda|_{\widetilde{T}})= \frac{3}{16}(1-e^{-1})\tau \sum_{U\in \mathcal{P}}\pi'(U).
\end{align*}
This concludes the proof of Lemma \ref{lemma:low_bound_OPT_1}. 
\end{proof}

\begin{lemma}
\label{lemma:low_bound_OPT_2}
Let $(T,\llambda)$ be a stochastic instance, $\mathcal{P}$ a balanced partition for this instance, $(T',\llambda^h)$ the corresponding heavy instance, and $\tau > 0$. 
It holds that 
\begin{equation*}
    \E\Big[\cost(\opt(\sigma), T)\Big] + \tau \cdot \sum_{U\in\mathcal{P}^*} \pi'(U) \ge \E\Big[\cost(\opt(\sigma'), T') \mid \sigma'\sim(T',\llambda^h)\Big],
\end{equation*}
where the expectation on the left side is taken over all random sequences $\sigma\sim (T,\llambda)^\tau$. 
\end{lemma}

\begin{proof}
Let $\sigma$ be an sequence of requests for $T$ of duration $\tau$ and let $S$ be a schedule for $\sigma$. Let $\sigma^h$ be the corresponding sequence for the heavy instance. 
We construct a schedule $S_h$ for $\sigma^h$ in $T'$ as follows. For each service $s\in S$, that serves requests located on $R(s)$, create a service in $S_h$ that serves the points $R(s)\cup\{z_U \mid U\in \mathcal{P}^*, U\cap R(s)\neq \emptyset\}$, with the same service time as $s$. 

It is clear that $S_h$ is a valid schedule for $\sigma^h$. 
It is also clear that $\delay(S,\sigma)=\delay(S_h,\sigma^h)$. We now claim that 
\begin{equation} \label{eq:to_prove_opt}
    \weight(S,\sigma) + \sum_{U \in \mathcal{P}^*}N(\sigma \vert_U)\cdot \frac{\pi'(U)}{\llambda(U)} \ge \weight(S_h, \sigma^h),
\end{equation}
where $N(\sigma \vert_U)$ is the number of requests in $\sigma \vert_U$, i.e., the number of request in $\sigma$ that are located at some vertex in $U$. 
Indeed, for a service $s\in S$, the weight of the rooted tree induced by $R(s) \cup \{z_U \mid U \in \mathcal{P}^*, U \cap R(s) \neq \emptyset\}$ is equal to the weight of the rooted tree induced by $R(s)$ plus the sum of the weights of the edges $(z_U, z'_U)$ for each $U \in \mathcal{P}^*$ such that $U\cap R(s) \ne \emptyset$. To see this, notice that for a service $s \in S$, and a part $U \in \mathcal{P}^*$ such that $U \cap R(s)$, the vertex $z'_U$ is contained in the rooted subtree served by $s$.  Now, the weight of $(z_U, z'_U)$ is equal to $\pi'(U)/\llambda(U)$, and for $U\in \mathcal{P}^*$, we have $U\cap R(s)\neq \emptyset$ if and only if $N(\sigma \vert_U) \ge 1$. This shows inequality \eqref{eq:to_prove_opt}. 

Note that $\cost(S, \sigma) = \delay(S, \sigma) + \weight(S, \sigma)$ and $\delay(S, \sigma) = \delay(S_h, \sigma^h)$.
Combining these two inequalities together with (\ref{eq:to_prove_opt}), we obtain
\begin{equation*}
    \cost(S,\sigma) + \sum_{U \in \mathcal{P}^*} N(\sigma \vert_U) \cdot \frac{\pi'(U)}{\llambda(U)} \ge \cost(S_h, \sigma^h)\ge \cost(\opt(\sigma^h)). 
\end{equation*}
In expectation, when $\sigma\sim(T,\llambda)$, this becomes, using linearity of expectation:
\begin{equation*}
    \E[\cost(\opt(\sigma),T)] + \sum_{U \in \mathcal{P}^*}\E[N(\sigma \vert_U)] \cdot \frac{\pi'(U)}{\llambda(U)} \ge \E[\cost(\opt(\sigma^h))].
\end{equation*}
Besides, the expected number of requests located in a part $U \in \mathcal{P}^*$ is $\E[N(\sigma \vert_U)] = \tau \cdot \llambda(U)$. 
Finally, thanks to Proposition \ref{proposition:correspondance_lambdas}, we know that for any random variable $X(\sigma')$ that depends on a sequence $\sigma'$ for $T'$ of duration $\tau$, we have $\E[X(\sigma^h) \mid \sigma \sim (T,\llambda)^\tau]=\E[X(\sigma') \mid \sigma' \sim (T',\llambda^h)^\tau]$. 
This finishes the proof of Lemma \ref{lemma:low_bound_OPT_2}.
\end{proof}

\subsubsection{Upper bound on the cost of GEN}
\label{section:upper_bound_gen}

\begin{lemma} \label{lemma:bound_GEN}
    Let $(T,\llambda)$ be a stochastic instance, $\mathcal{P}$ the balanced partition computed by the algorithm \gen, $(T',\llambda^h)$ the corresponding heavy instance, and $\tau>0$. It holds that
    \begin{equation*}
        \E[\cost(\general(\sigma), T)] \le \E\left[\cost(\plan(\sigma'), T')\mid \sigma'\sim(T',\llambda^h)^\tau\right] + \tau \cdot \sum_{U\in\mathcal{P}} \pi'(U),
    \end{equation*}
    where the expectation on the left side is taken over all sequences $\sigma \sim (T,\llambda)^\tau$. 
\end{lemma}

\begin{proof}
Given $\sigma \sim (T, \llambda)^\tau$, let $S_I \subseteq \gen(\sigma)$ be the set of services to corresponds of applying $\instant$ during the execution of $\gen$ on $\sigma$. 
Let $S' = \gen(\sigma) \setminus S_I$ be the set of services ordered by $\plan$. 
It is clear that $\cost(\gen(\sigma)) = \cost(S_I) + \cost(S')$.  

To prove the lemma, we show the two following bounds: 
\begin{equation} \label{eq:gen_light}
    \E[\cost(S_I)] = \tau \cdot \pi'(\gamma(\mathcal{P})),
\end{equation}
\begin{equation} \label{eq:gen_heavy}
    \E[\cost(S')] \le \E[\cost(\plan(\sigma^h), T') \mid \sigma' \sim (T',\llambda^h)^\tau] + \tau \cdot \sum_{U\in\mathcal{P}^*} \pi'(U),
\end{equation}
where the expectations are taken over all sequences $\sigma \sim (T,\llambda)^\tau$. 
It is easy to see that these two results implies the statement of the Lemma. 

\paragraph{Proof of \eqref{eq:gen_light}.}
By definition of $\pi'(\cdot)$, we have $\pi'(\gamma(\mathcal{P})) = \pi(\gamma(\mathcal{P}))$. 
Recall that $\pi(\gamma(\mathcal{P})) = \sum_{u\in \gamma(\mathcal{P})} \llambda(u) \cdot d(u,\gamma(T))$ is the expected distance from the root to the location of the requests in $\sigma \sim (T[\gamma(\mathcal{P})], \llambda|_{\mathcal{P}})$, which in turn is the expected weight cost of the schedule produced by \instant. 
Since this algorithm serves each request in $\gamma(\mathcal{P})$ immediately when it arrives, the delay cost is 0. Hence, we obtain \eqref{eq:gen_light}. 

\paragraph{Proof of \eqref{eq:gen_heavy}.}
Let $\sigma^h$ be the heavy sequence for $T'$ associated with $\sigma$. We show that
\begin{equation}
    \cost(S', T) \le \cost(\plan(\sigma^h), T') + \sum_{U \in \mathcal{P}^*} \sum_{r \in \sigma \vert_U} d(\ell(r), z'_U), 
    \label{eq:to_prove_gen}
\end{equation}

It is easy to check that $\delay(S', T) = \delay(\plan(\sigma^h), T')$. 
Now we claim that 
\begin{equation} \label{eq:to_prove_gen_2}
    \weight(S', T) \le \weight(\plan(\sigma^h), T') + \sum_{U \in \mathcal{P}^*} \sum_{r \in \sigma\vert_U} d(\ell(r),z'_U),  
\end{equation}
which together with the previous equation on the delay cost implies \eqref{eq:to_prove_gen}. 
The only thing left now is to prove \eqref{eq:to_prove_gen_2}. 

Let $s_h$ be a service in $\plan(\sigma^h)$ and let $s_g$ denote the corresponding service in $S'$. 
Let $U \in \mathcal{P}^*$ such that $z_U \in R(s_h)$. Let $t' < t(s_h)$ be the latest time at which $\plan(\sigma^h)$ served $z_U$. 
It is clear by the definition of GEN that all requests $(t, u) \in \sigma$ with $u \in U$ and $t \le t'$ have been served by GEN at time $t'$, and that all requests $(t, u) \in \sigma$ with $u \in U$ and $t' < t \le t(s_h)$ are served by $s_g$. 

Let $Q_h$ denote the subtree in $T$ served by $s_h$ and let $Q_g$ be the subtree of $T'$ served by $s_g$. 
For a request $r \in \sigma$ located in some $U\in \mathcal{P}^*$, let $P_r$ denote the path in $T'$ from $\ell(r)$ to $z'_U$. 
If $z_U \in Q_h$, then it is necessarily $z'_U \in Q_h$, so $Q_h \cup P_r$ is connected. 
Thus, the subtree $$Q'_g := Q_h \bigcup \left(\bigcup_{U \in \mathcal{P}^*, z_U \in R(s_h)} \bigcup_{r \in R(s_g), \ell(r) \in U} P_r\right)$$ contains all the requests served by $s_g$, which implies that 
\begin{equation*}
    w(Q_g) \le w(Q'_g)\le w(Q_h) + \sum_{U \in \mathcal{P}^*, z_U\in R(s_h)} \sum_{r\in R(s_g), \ell(r)\in U} d(\ell(r),z'_U).
\end{equation*}
We obtain \eqref{eq:to_prove_gen_2} by summing this inequality over all services $\sigma^h \in \plan(\sigma^h)$. Hence, we have proved \eqref{eq:to_prove_gen}. 

To obtain \eqref{eq:gen_heavy} from \eqref{eq:to_prove_gen}, we take expectation over all random sequences $\sigma \sim (T,\llambda)^\tau$. 
First, by Proposition \ref{proposition:correspondance_lambdas}, we know that 
\begin{equation*}
    \E[\cost(\plan(\sigma^h),T') \mid \sigma^h \sim (T',\llambda^h)^\tau] = \E[\cost(\plan(\sigma),T') \mid \sigma \sim (T,\llambda)^\tau].
\end{equation*}
Finally, in expectation we have    
\begin{align*}
    \E\bigg[\sum_{U \in \mathcal{P}^*} \sum_{r \in \sigma \vert_U} d(\ell(r), z'_U)\bigg] 
    &= \sum_{U \in \mathcal{P}^*}\sum_{u \in U} \E[\text{number of requests of }\sigma\text{ located on $u$}]\cdot d(u, z'_U)\\
    &= \sum_{U \in \mathcal{P}^*} \sum_{u \in U} (\llambda(u) \cdot \tau) \cdot d(u, z'_U) = \tau \cdot \sum_{U \in \mathcal{P}^*} \pi'(U). 
\end{align*}
This concludes the proof of \eqref{eq:gen_heavy}, and thus the proof of the lemma.
\end{proof}

\subsubsection{Proof of Theorem \ref{theorem:general_any}}
\label{section:proof_main_theorem}

Let $(T, \llambda)$ be a stochastic instance and let $\mathcal{P}$ be the balanced partition of $T$. Let $(T', \llambda^h)$ denote the corresponding heavy instance, and let $\tau > 0$. 
Taking expectation over all the sequences $\sigma\sim(T,\llambda)^\tau$, we obtain:
\begin{align*}
    \lefteqn{\E_\sigma^\tau[\cost(\general(\sigma), T)]} \\
    &\le \E\Big[\cost(\plan(\sigma'), T') \mid \sigma'\sim(T',\llambda^h)^\tau\Big] + \tau  \sum_{U \in \mathcal{P}} \pi'(U) && \text{(Lemma \ref{lemma:bound_GEN})} \\
    &\le \frac{64}{3} \cdot \E\Big[\cost(\opt(\sigma'), T') \mid \sigma'\sim(T',\llambda^h)^\tau\Big] + \tau \sum_{U \in \mathcal{P}} \pi'(U) && \text{(Thm. \ref{theorem:plan_heavy} and Prop. \ref{proposition:heavy_instance_is_heavy})} \\
    &\le \frac{64}{3}\left(\E\big[\cost(\opt(\sigma), T)\big]+\tau\sum_{U\in\mathcal{P}^*} \pi'(U) \right) + \tau \sum_{U \in \mathcal{P}} \pi'(U) && \text{(Lemma \ref{lemma:low_bound_OPT_2})}\\ 
    &\le \frac{64}{3} \cdot \E\big[\cost(\opt(\sigma), T)\big] + \left(\frac{64}{3}+1\right) \left(\tau \cdot \sum_{U \in \mathcal{P}} \pi'(U)\right) &&  \\
    &\le \frac{64}{3} \cdot \E\big[\cost(\opt(\sigma), T)\big] + \frac{67}{3}\cdot\frac{16}{3(1-e^{-1})} \cdot \E\big[\cost(\opt(\sigma), T)\big] && \text{(Lemma \ref{lemma:low_bound_OPT_1})}\\ 
    &< 210 \cdot \E\big[\cost(\opt(\sigma), T)\big]. 
     &&
\end{align*}
This proves the theorem. 

\section{Other related works}
\label{section:related}
The MLA problem was first introduced by Bienkowski et al.~\cite{bienkowski2016online} and they study a more general version in their paper, where the cost of delaying a request $r$ by a duration $t$ is $f_r(t)$.
Here $f_r(\cdot)$ denotes the delay cost function of $r$ and it only needs to be non-decreasing and satisfy $f_r(0) = 0$.
An O($d^4 2^d$)-competitive online algorithm is proposed for this general delay cost version problem, where $d$ denotes the depth of the given tree. 
Besides, a deadline version of MLA is also considered in \cite{bienkowski2016online}, where each request $r$ has a time window (between its arrival and its deadline) and it has to be served no later than its deadline. 
The target is to minimize the total service cost for serving all the requests.
For this deadline version problem, they proposed an online algorithm with better competitive ratio of $d^2 2^d$. 
Later, the competitiveness of MLA are further improved to O($d^2$) by Azar and Touitou \cite{azar2019general} for the general delay cost version and to O($d$) by Buchbinder et al. \cite{buchbinder2017depth} for the deadline version.\footnote{Later, Mcmahan \cite{mcmahan2021d} further improve the competitive ratio to $d$ for the deadline version of MLA.} 
However, for the delay cost version of MLA, no matching lower bound has been found thus far --- the current best lower bound on MLA (with delays) is only 4 \cite{bienkowski2016online, bienkowski2020online, bienkowski2021new}, restricted to a path case with linear delays. 
In the offline setting, MLA is NP-hard in both delay and deadline versions \cite{arkin1989computational, becchetti2009latency} and a 2-approximation algorithm was proposed by Becchetti et al. \cite{becchetti2009latency} for the deadline version.
For a special path case of MLA with the linear delay, Bienkowski et al. \cite{bienkowski2013online} proved that the competitiveness is between 3.618 and 5, improving on an earlier 8-competitive algorithm given by Brito et al. \cite{brito2012competitive}.
Thus far, no previous work has studied MLA in the stochastic input model, no matter the delay or deadline versions.

Two special cases of MLA with linear delays, one TCP-acknowledgment (equivalent to MLA with tree being an edge, i.e. $d = 1$) and one Joint Replenishment (abbv. JRP, equivalent to MLA with tree being a star, i.e. $d = 2$) are of particular interests. 
This is because, TCP-acknowledgment (a.k.a. single item lot-sizing problem in operation research community, see e.g. \cite{brahimi2006single, jans2008modeling, quadt2008capacitated, bushuev2015review, karimi2003capacitated}) models the data transmission issue from sensor networks (see e.g. \cite{yuan2003synchronization, leung2007overview}), while JRP models the inventory control issue from supply chain management (see e.g. \cite{askoy1988multi, goyal1989joint, joneja1990joint, sindhuchao2005integrated, khouja2008review}). 
For TCP-acknowledgment, in the online setting there exists an optimal 2-competitive deterministic algorithm \cite{dooly2001line} and an optimal $\frac{1}{1 - e^{-1}}$-competitive randomized algorithm \cite{karlin2001dynamic, seiden2000guessing}; in the offline setting, the problem can be solved in O($n \log n$) time, where $n$ denotes the number of requests \cite{aggarwal1993improved}. 
For JRP, there exists a 3-competitive online algorithm based on primal-dual method proposed by Buchbinder et al. \cite{buchbinder2008online}, and no online algorithm achieves competitive ratio less than 2.754 \cite{bienkowski2014better}.
In the offline setting, JRP is NP-hard \cite{arkin1989computational} and also APX-hard \cite{nonner2009approximating, bienkowski2015approximation}. 
The current best approximation ratio for JRP is 1.791 (Bienkowski et al. \cite{bienkowski2014better}), improving on earlier results given by Levi et al. \cite{levi2004primal, levi2006improved, levi2008constant}.  
For a deadline version of JRP, Bienkowski et al. \cite{bienkowski2014better} proposed an optimal 2-competitive online algorithm. 
For the stochastic version, unfortunately, to our best extent, we did not find previous works on these two problems with requests following Poisson arrival model from a theoretical perspective, i.e., proposing online algorithm with their performances evaluated using RoE. 

Another problem, called online service with delays (OSD), first introduced by Azar et al. \cite{azar2017online}, is closely related to MLA (with linear delays). 
In this OSD problem, a $n$-points metric space is given as input.
The requests arrive at metric points over time and a server is available to serve the requests. 
The target is to serve all the requests in an online manner such that their total delay cost plus the total distance travelled by the server is minimized.
Note that MLA can be seen as a special case of OSD when the given metric is a tree and the server has to always come back to a particular tree vertex immediately after serving some requests elsewhere. 
For OSD, Azar et al. \cite{azar2017online} proposed a O($\log^4 n$)-competitive online algorithm in their paper.
Later, the competitive ratio for OSD is improved from $O(\log^2 n)$ (by Azar and Touitou \cite{azar2019general}) to $O(\log n)$ (by Touitou \cite{touitou2023improved}).

We remark here that besides MLA and OSD, many other online problems with delays/deadline have also drawn a lot of attentions recently, such as online matching with delays \cite{emek2016online, azar2017polylogarithmic, ashlagi2017min, bienkowski2017match, bienkowski2018primal, emek2019minimum, azar2020deterministic, deryckere2023online, liu2018impatient, azar2021min, melnyk2021online, mathieu2023online, kuo2024online}, facility location with delays/deadline \cite{bienkowski2021online, azar2019general, azar2020beyond}, Steiner tree with delays/deadline \cite{azar2020beyond}, bin packing with delays \cite{azar2019price, epstein2021bin, epstein2022open, ahlroth2013online}, set cover with delays \cite{azar2020set, touitou2021nearly, le2023power}, paging with delays/deadline \cite{gupta2020caching, gupta2022hitting}, list update with delays/deadline \cite{azar2023list}, and many others \cite{melnyk2021online, chen2022online, touitou2023frameworks, im2023online, kakimura2023deterministic, kawase2024online}.

\section{Concluding remarks}
\label{section:conclusion}
In this paper, we studied MLA with additional stochastic assumptions on the sequence of the input requests. 
In the case where the requests follow a Poisson arrival process, we presented a deterministic online algorithm with constant RoE. 
In the following text, we briefly discuss some potential future directions.

\paragraph{Does the greedy algorithm achieve a constant RoE?}
An intuitive heuristic algorithm for MLA is {\em Greedy}, which works as follows: 
{\em each time when a set of requests $R$ arriving at vertices $U \subseteq V(T)$ have the total delay cost equal to the weight of the minimal subtree of $T$ including $\gamma$ and $U$, serve all the requests $R$.}
Does this greedy algorithm achieves a constant RoE?

\paragraph{Generalize MLA with edge capacity and $k$ tree roots.} 
One practical scenario on MLA is that each edge has a capacity on the maximum number of requests served in one service if this edge is used, such as \cite{quadt2008capacitated, karimi2003capacitated, sindhuchao2005integrated}. We conjecture that some O(1)-RoE online algorithm can be proposed for this generalized MLA with edge capacity. 
Another generalized version of MLA is to assume $k$ tree roots available for serving requests concurrently. 
That is, a set of pending requests can be served together by connecting to any of $k$ servers. 
The question is, how to design an online algorithm for this $k$-MLA problem?
Does there exist O(1)-RoE algorithm still?

\paragraph{Other online network design problems with delays in the Poisson arrival model. }
Recall that the online problems of service with delays (and its generalization called $k$-services with delays), facility location with delays, Steiner tree/forest with delays are all closely related to MLA with delays. 
Does there exist online algorithm with O(1)-RoE for each problem?

\section*{Acknowledgement}
This work is supported by ERC CoG grant TUgbOAT 772346, NCN grant 2020/37/B/ST6/04179 and NCN grant 2022/45/B/ST6/00559.

\bibliographystyle{siam}
\bibliography{bibliography}

\appendix
\section{Missing Proofs in Section \ref{section:preliminaries}} 
\label{appendix:preliminaries}
We first introduce the two well-known properties of the exponential distribution, which will be used to prove Proposition \ref{prop:poisson_independence_under_taking_subsegment}, Proposition \ref{prop:poisson_independence_under_taking_subsegment} and Proposition \ref{prop:poisson_model_centralized}. 

\begin{proposition}[memoryless property] \label{proposition:memoryless}
    If $X$ is an exponential variable with parameter $\lambda$, then for all $s,t \ge 0$, we have
    \begin{equation*}
        \P(X > s + t \mid  X > s) = \P(X > t) = e^{- \lambda t}.
    \end{equation*}
\end{proposition}

\begin{proposition} \label{proposition:minimum}
    Given $n$ independent exponential variables $X_i \sim \expdistr{\lambda_i}$ for $i \in [n]$, let $Z := \min\{X_1, X_2, \dots, X_n\}$ and let $\lambda := \sum_{i = 1}^n \lambda_i$. It holds that
    \begin{tasks}[style=enumerate](3)
        \task $Z \sim \expdistr{\lambda}$,
        \task $\P(Z = X_i) =  \lambda_i / \lambda$,
        \task $Z \perp \{Z = X_i\}$,
    \end{tasks}
    where $\perp$ denotes independence.
\end{proposition}

To illustrate the equivalence between the distributed Poisson arrival model (Definition \ref{def:poisson_model_distributed}) and the centralized Poisson arrival model (Proposition \ref{prop:poisson_model_centralized}), see Figures \ref{figure:distributed_Poisson} and \ref{figure:centralized_Poisson} for visual supports.

\begin{figure}[ht]
\input{figures/distributed_poisson}
\vspace{-16pt}
\caption{Example showing the correspondence between distributed Poisson arrival model and exponential timers without resets. $Y_x^i$ and $T_x^{x, i}$ both represent the waiting time between the $(i-1)$-th and the $i$-th request arriving at any vertex $x \in V(T)$. The graph on the right highlights in blue the waiting times between the consecutive arrivals from the perspective of the tree vertices.}
\label{figure:distributed_Poisson}
\end{figure}

\begin{figure}[ht]
\input{figures/centralized_poisson}
\vspace{-16pt}
\caption{Example showing the correspondence between the centralized Poisson arrival model and timers with resets. On the left, the waiting time between appearances of the $i$-th and the $(i+1)$-th requests is determined by the variable $Y_{i+1} \sim \expdistr{\llambda(T)}$. On the right, a double-headed arrow represents a timer, while a single-headed arrow means that we had to reset a timer.}
\label{figure:centralized_Poisson}
\end{figure}

\subsection{Proof of Proposition \ref{prop:poisson_model_centralized}}
Consider the distributed model where for each vertex $u \in V(T)$, we define an exponential variable $Y_u^1$ representing the time of arrival of the first request located at $u$. 
If we look at the whole tree, the time of arrival of the first request $r$ is determined by the minimum of all these variables, $\min_{u \in V(T)} Y_u^1$. 
We denote this variable by $Y_1$. 
By Proposition \ref{proposition:minimum}, we know that $Y_1$ follows an exponential distribution with parameter $\llambda(T)$ being the sum of components' parameters. 
Moreover, by the second property presented in this proposition, we know that the probability of $r$ arriving at vertex $u$ equals $\llambda(u) / \llambda(T)$ for each $u \in V(T)$.

At time $Y_1$ when the first request arrives, we associate each vertex $u \neq l(r)$ with a new independent exponential random variable $Z_u^1 \sim \expdistr{\llambda(u)}$. 
By the memoryless property from Proposition \ref{proposition:memoryless}, for each vertex $u$ the arrival time determined by $t(r) + Z_v^1$ follows the same distribution as $Y_u^1$ conditioned on being greater than $t(r)$. 
This shows that we can look at the first request arrival as it was defined by the centralized model and the consequent requests still follow the distributed model. 
We can continue this process to transform the distributed model into a centralized one.

\subsection{Proof of Proposition \ref{prop:poisson_independence_under_taking_subsegment}}
We only need to prove that, given $\sigma_1 \sim (T, \llambda)^{\tau_1}$ and $\sigma_2 \sim (T, \llambda)^{\tau_2}$, denoting by $\sigma$ a merged sequence of $\sigma_1$ and $\sigma_2$ (each request in $\sigma_2$ has its arrival time delayed by a duration $\tau_1$), we have $\sigma \sim (T, \llambda)^{\tau_1 + \tau_2}$. 
Indeed, given any sequence $\sigma_2 \sim (T, \llambda)^{\tau_2}$ and increasing the arrival time of each request by $\tau_1$, we obtain a sequence $\sigma'_2$ generated according to Poisson arrivals during time interval $[\tau_1, \tau_1 + \tau_2]$. 
Let the arrival time of the last request in $\sigma_1 \sim (T, \llambda)^{\tau_1}$ be $t$ and we thus know that no request is generated during time interval $[t, \tau_1]$. 
By the memoryless property from Proposition \ref{proposition:memoryless}, we thus know that $\sigma'_2$ also follows Poisson arrivals during time interval $[t, \tau_1 + \tau_2]$.
As a result, by Definition \ref{def:poisson_model_distributed}, we can conclude that $\sigma = \sigma_1\sigma'_2$ follows Poisson arrival model during time interval $[0, \tau_1 + \tau_2]$ and hence this proposition.  

\subsection{Proof of Proposition \ref{prop:poisson_fixed_number_uniform_arrivals}}
To prove this proposition, we show that the joint density function of the $n$ requests' arrival times is identical to the joint density function of the order statistics corresponding to $n$ independent random variables uniformly distributed over $[0, \tau]$. 
Let $X_1$, $X_2$, $\dots$, $X_n$ denote the $n$ independent identical variables, each with the same density function $f(x)$. 
Let $X_{(1)}$, $X_{(2)}$, $\dots$, $X_{(n)}$ denote these $n$ variables in an increasing order, i.e., for each $i \in [n]$, $X_{(i)}$ is the $i$-th smallest variable among $X_1, X_2, \dots, X_n$.
Denoting by $x_1 < x_2 < \dots < x_n$, then the joint density of $X_{(1)}$, $X_{(2)}$, $\dots$, $X_{(n)}$ is
\begin{equation*}
    f(x_1, x_2, \dots, x_n) = n! \prod_{i = 1}^n f(x_i).
\end{equation*}
As a result, given $n$ independent variables $U_1$, $U_2$, $\dots$, $U_n$, each of which is uniformly drawn from $[0, \tau]$ and hence has a density function of $g(t) = \frac{1}{\tau}$, the density function of $U_{(1)}$, $U_{(2)}$, $\dots$, $U_{(n)}$ is
\begin{equation*}
    g(t_1, \dots, t_n) = \frac{n!}{\tau^n}.
\end{equation*}

Now, let $W_i$ denote the waiting time between the ($i-1$)-th request's arrival and the $i$-th request's arrival, which follows an exponential distribution $\expdistr{\lambda}$. 
Let $S_i = \sum_{j = 1}^i W_i$ denote the arrival time of the $i$-th request, according to Poisson arrival model. 
Now we derive the joint density function of $S_{(1)}, \dots, S_{(n)}$ (denoted by $f(t_1, \dots, t_n)$), given that $n$ requests are generated within $[0, \tau]$ (i.e., $N(\sigma) = n$).
Given $0 < t_1 < t_2 < \dots < t_n < t_{n+1} = \tau$, for each $i \in [n]$, let $h_i > 0$ be a value small enough so that $t_i + h_i < t_{i+1}$. 
By definition of Poisson arrival process, given that the Poisson arrival rate $\lambda$, the probability to generate $n \ge 0$ requests during any interval of length $\ell$ is 
equal to 
\begin{equation*}
    e^{- \lambda \ell} \cdot \frac{(\lambda \ell)^n}{n!}.
\end{equation*}
Therefore, we have
\begin{eqnarray}
    \lefteqn{\P(t_i \le S_i \le t_i + h_i \text{ for each } i \in [n] \,|\, N(\sigma) = n)} \nonumber \\
    &=& \frac{\prod_{i = 1}^n \P(\text{exactly 1 request generated in } [t_i, t_i + h_i])}{\P(N(\sigma) = n)} \cdot \P(\text{no request arrives elsewhere in } [0, \tau]) \nonumber \\
    &=& \frac{\prod_{i = 1}^n \lambda h_i \cdot e^{-\lambda h_i}}{e^{- \lambda \tau} \cdot \frac{(\lambda \tau)^n}{n!}}  \cdot e^{\lambda (\tau - h_1 - \dots - h_n)} = \frac{n!}{\tau^n} \cdot h_1 \cdot h_2 \cdot \dots \cdot h_n. \nonumber 
\end{eqnarray}
As a result, 
\begin{equation*}
    \frac{\P(t_i \le S_i \le t_i + h_i \text{ for each } i \in [n] \,|\, N(\tau) = n)}{ h_1 \cdot h_2 \cdot \dots \cdot h_n} = \frac{n!}{\tau^n}.
\end{equation*}
By letting $h_i \to 0$ for each $i$, the conditional joint density of $S_1, \dots, S_n$ given $N(\tau) = n$ becomes 
\begin{equation*}
    f(t_1, \dots, t_n) = \frac{n!}{\tau^n}.
\end{equation*}
Since $f(t_1, \dots, t_n) = g(t_1, \dots, t_n)$, we thus have Proposition \ref{prop:poisson_fixed_number_uniform_arrivals}.

\subsection{Proof of Proposition \ref{prop:poisson_total_delay_cost}}
Again, let $W_i$ denote the waiting time between the ($i-1$)-th request's arrival time and the $i$-th request's arrival time, which follows an exponential distribution $\expdistr{\lambda}$. 
Let $S_i = \sum_{j = 1}^i W_i$ denote the $i$-th request's arrival time. 
We first calculate the expected total delay cost under the condition that $N(\sigma) = n$ requests are generated within $[0, \tau]$. 
\begin{equation*}
    \E_{\sigma}^{\tau}\left[\sum_{i = 1}^{N(\sigma)} (\tau - S_i) \,|\, N(\sigma) = n\right]
    = n \tau - \E_{\sigma}^{\tau}\left[\sum_{i = 1}^{N(\sigma)} S_i \,|\, N(\sigma) = n\right].
\end{equation*}
Let $U_1$, $U_2$, $\dots$, $U_n$ denote the $n$ uniform variables drawn from $[0, \tau]$ and let $U_{(1)}$, $U_{(2)}$, $\dots$, $U_{(n)}$ denote these variables in an increasing order. 
We thus have $\sum_{i = 1}^n U_i = \sum_{i = 1}^n U_{(i)}$ and
\begin{equation*}
    \E\left[\sum_{i = 1}^n U_i\right] = \E\left[\sum_{i = 1}^n U_{(i)}\right].
\end{equation*}
By Proposition \ref{prop:poisson_fixed_number_uniform_arrivals}, we further have 
\begin{equation*}
    \E_{\sigma}^{\tau}\left[\sum_{i = 1}^{N(\sigma)} S_i \,|\, N(\sigma) = n\right] = \E\left[\sum_{i = 1}^n U_{(i)}\right].
\end{equation*}
As a result, we have
\begin{equation*}
    \E_{\sigma}^{\tau}\left[\sum_{i = 1}^{N(\sigma)} (\tau - S_i) \,|\, N(\sigma) = n\right]
    = n \tau - \E\left[\sum_{i = 1}^n U_i\right] = n \tau - \sum_{i = 1}^n \E[U_i] = n \tau - \frac{n \tau}{2} = \frac{n \tau}{2}
\end{equation*}
and hence,
\begin{equation*}
    \E_{\sigma}^{\tau}\left[\sum_{i = 1}^{N(\sigma)} (\tau - S_i)\right] = \frac{\tau}{2} \cdot \E_{\sigma}^{\tau}[N(\sigma)] = \frac{\tau}{2} \cdot \llambda(T) \cdot \tau = \frac{1}{2} \cdot \llambda(T) \cdot \tau^2.
\end{equation*}

\section{Missing contexts in Section \ref{section:overview}} 
\label{appendix:overview}

\begin{lemma}
    There exists a stochastic instance ($T, \llambda$) such that both $\frac{\E\left[\cost(\instant(\sigma)) \mid \sigma \sim (T, \llambda)^\tau\right]}{\E\left[\cost(\opt(\sigma)) \mid \sigma \sim (T, \llambda)^\tau\right]}$ and $\frac{\E\left[\cost(\plan(\sigma)) \mid \sigma \sim (T, \llambda)^\tau\right]}{\E\left[\cost(\opt(\sigma)) \mid \sigma \sim (T, \llambda)^\tau\right]}$ are unbounded. 
\end{lemma}
Consider such an MLA instance ($T, \llambda$) where $T$ has a depth of 2. 
The root vertex $\gamma$ has only one child $u$ with $\llambda(u) = 0$ and $u$ has $n$ child vertices $v_1, \dots, v_n$, with $\llambda(v_i) = \frac{1}{n}$ for each $i \in [n]$. 
The edge ($\gamma, u$) has a weight of $\sqrt{n}$ and each edge ($u, v_i$) has a weight of 1. 
By definition, such an instance is neither light nor heavy.
Note that if INSTANT is applied to deal with this instance, then to serve each request, only a weight cost equal to $\sqrt{n} + 1$ is incurred and hence the expected cost of INSTANT is
\begin{equation*}
    (\sqrt{n} + 1) \cdot \E_{\sigma}^{\tau}[N(\sigma)] = (\sqrt{n} + 1) \cdot \llambda(T) \cdot \tau = \Theta(\sqrt{n}) \cdot \tau.
\end{equation*}
If PLAN algorithm is applied, then the period for each vertex $v_i$ is determined as
\begin{equation*}
    \sqrt{\frac{2 \cdot (\sqrt{n} + n \cdot 1)}{\llambda(T)}} = \sqrt{2(n + \sqrt{n})}. 
\end{equation*}
For each period of length $\sqrt{2(n + \sqrt{n})}$, a weight cost of $\sqrt{n} + n$ is incurred (since the whole tree $T$ is bought) and the expected delay cost produced is $\frac{1}{2} \cdot \llambda(T) \cdot (\sqrt{2(n + \sqrt{n})})^2 = n + \sqrt{n}$ also. 
The expected cost PLAN is thus equal to 
\begin{equation*}
    \frac{\tau}{\sqrt{2(n + \sqrt{n})}} \cdot 2(n + \sqrt{n}) = \sqrt{2(n + \sqrt{n})} \cdot \tau = \Theta(\sqrt{n}) \cdot \tau.
\end{equation*}
However, consider the following algorithm ALG:
\begin{itemize}
    \item[-] the edge ($\gamma, u$) is bought periodically with period being $n^{\frac{1}{4}}$;
    \item[-] if at one request located at $v_i$ is pending at time $j$, serve this request at this moment. 
\end{itemize}
In this way, for each period of length $n^{\frac{1}{4}}$, 
\begin{itemize}
    \item[-] the expected number of requests generated within this period is $\llambda(T) \cdot n^{\frac{1}{4}} = n^{\frac{1}{4}}$ --- the expected weight cost is thus equal to $\sqrt{n} + 1 \cdot n^{\frac{1}{4}} = \sqrt{n} + n^{\frac{1}{4}}$;
    \item[-] the expected delay cost is $\frac{1}{2} \cdot \llambda(T) \cdot (n^{\frac{1}{4}})^2 = \frac{\sqrt{n}}{2}$. 
\end{itemize}
As a result, the expected cost produced in each period is equal to $\sqrt{n} + n^{\frac{1}{4}} + \frac{\sqrt{n}}{2} = 1.5\sqrt{n} + n^{\frac{1}{4}}$
and the expected cost of this algorithm is equal to 
\begin{equation*}
    \frac{\tau}{n^{\frac{1}{4}}} \cdot (1.5 \sqrt{n} + n^{\frac{1}{4}}) = (1.5 \cdot n^{\frac{1}{4}} + 1) \cdot \tau = \Theta(n^{\frac{1}{4}}) \cdot \tau.
\end{equation*}
Notice that
\begin{equation*}
    \overline{\lim_{\tau \to \infty}} \frac{\E_{\sigma}^\tau[\cost(\instant(\sigma))]}{\E_{\sigma}^\tau[\cost(\alg(\sigma))]} = \Theta(n^{\frac{1}{4}}) \text{ and }
    \overline{\lim_{\tau \to \infty}} \frac{\E_{\sigma}^\tau[\cost(\plan(\sigma))]}{\E_{\sigma}^\tau[\cost(\alg(\sigma))]} = \Theta(n^{\frac{1}{4}}). 
\end{equation*}
By letting $n \to \infty$, we can conclude that both INSTANT and PLAN achieve unbounded RoEs. 
\end{document}